\documentclass[11pt,onecolumn]{IEEEtran}

\usepackage{etoolbox}
\usepackage{cite}
\newcommand{\FigDat}[2]{
\ifstrequal{#1}{SIR_Hist}{
\begin{figure}[t]
        \begin{center}\includegraphics[scale=#2]{fig_SIR_hist}\end{center}
    \caption{Histogram of the normalized SIR, \eqref{d:Normalized SIR}, using a cluster of the $K=4$ nearest BSs and varying number of antennas per BS ($L$). The figure also shows the normalized SIR predicted by Lemma \ref{Lemma:SIR}.}
    \label{f:SIR_Hist_fig}
\end{figure}
}{}
\ifstrequal{#1}{SIR_Line}{
\begin{figure}[t]
        \begin{center}\includegraphics[scale=#2]{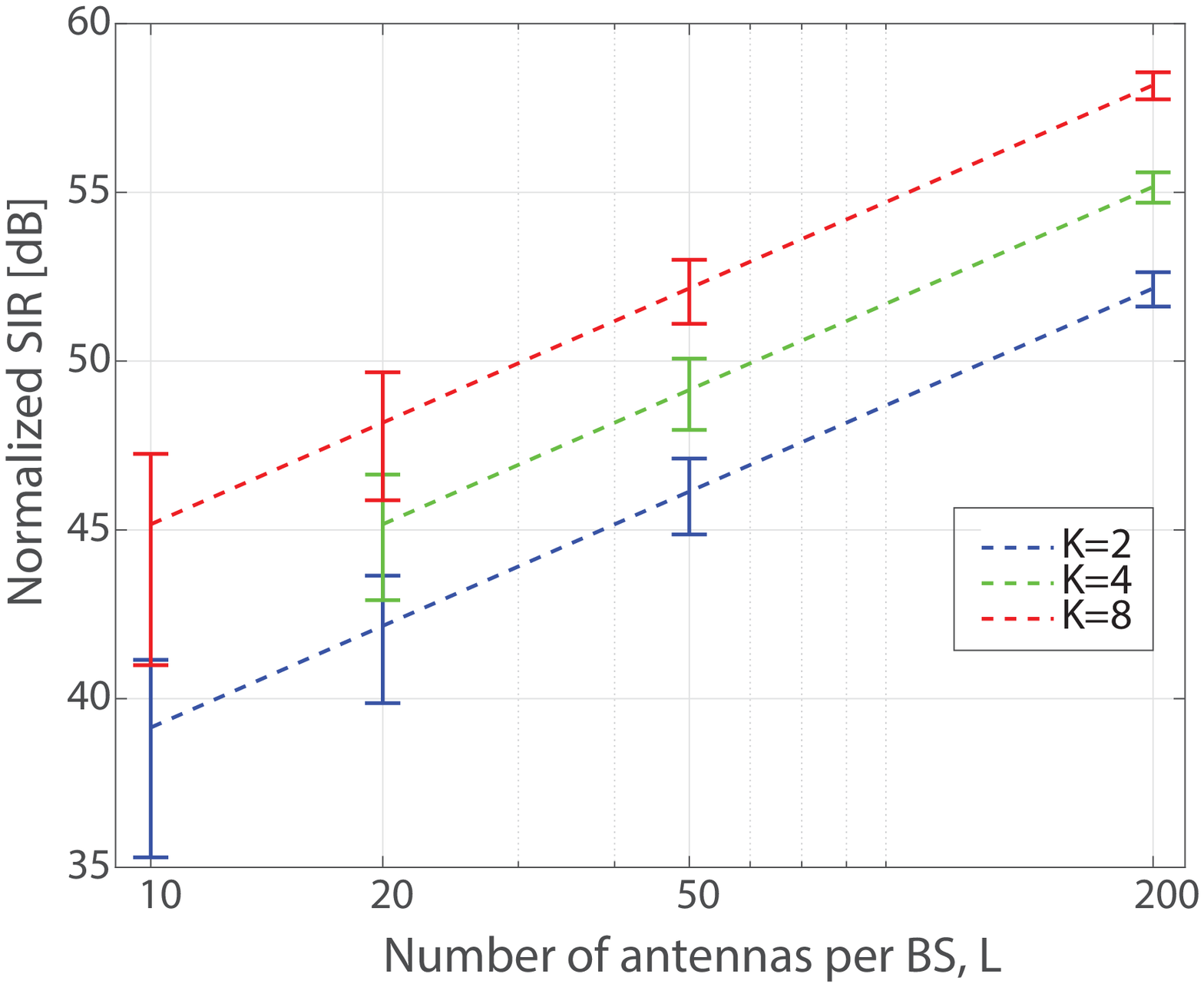}\end{center}
    \caption{The normalized SIR, \eqref{d:Normalized SIR}, as a function of the number of antennas per BS for various cluster sizes ($K$). The error-bars show the $\pm 1$ Std through MC simulations while dashed lines show the predicted values by Theorem \ref{Lemma:SIR}.}
    \label{f:SIR_Line_fig}
\end{figure}
}{}
\ifstrequal{#1}{fig_rate}{
\begin{figure}[t]
        \begin{center}\includegraphics[scale=#2]{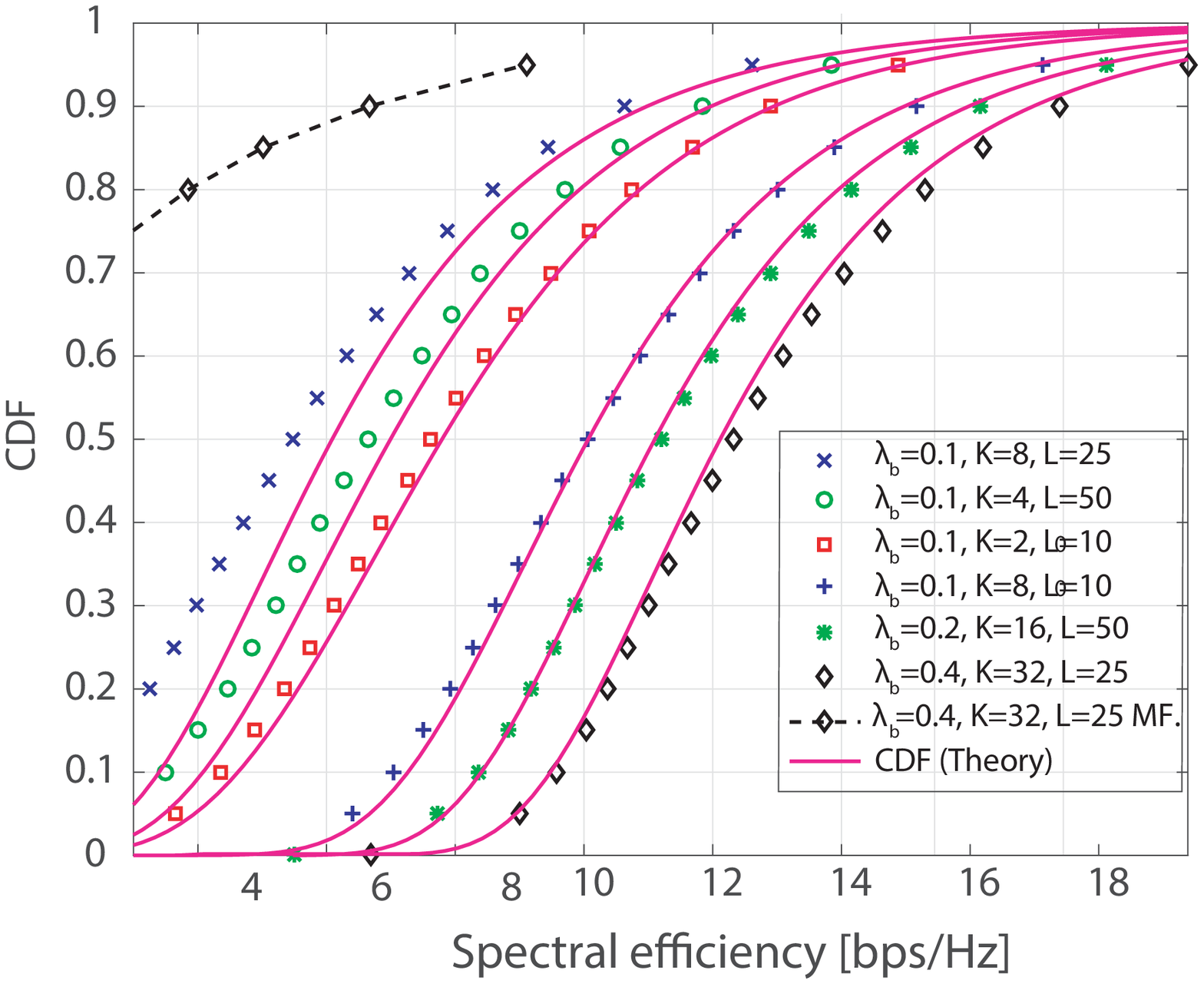}\end{center}
    \caption{The CDF of the spectral efficiency in a network with varying parameters. Solid lines depict the theoretical result of Theorem \ref{Theorem:Main}. Markers depict simulation results. The dashed line indicates the CDF of a matched-filter receiver for reference.}
    \label{f:rate_fig}
\end{figure}
}{}

\ifstrequal{#1}{fig_alpha}{
\begin{figure}[t]
        \begin{center}\includegraphics[scale=#2]{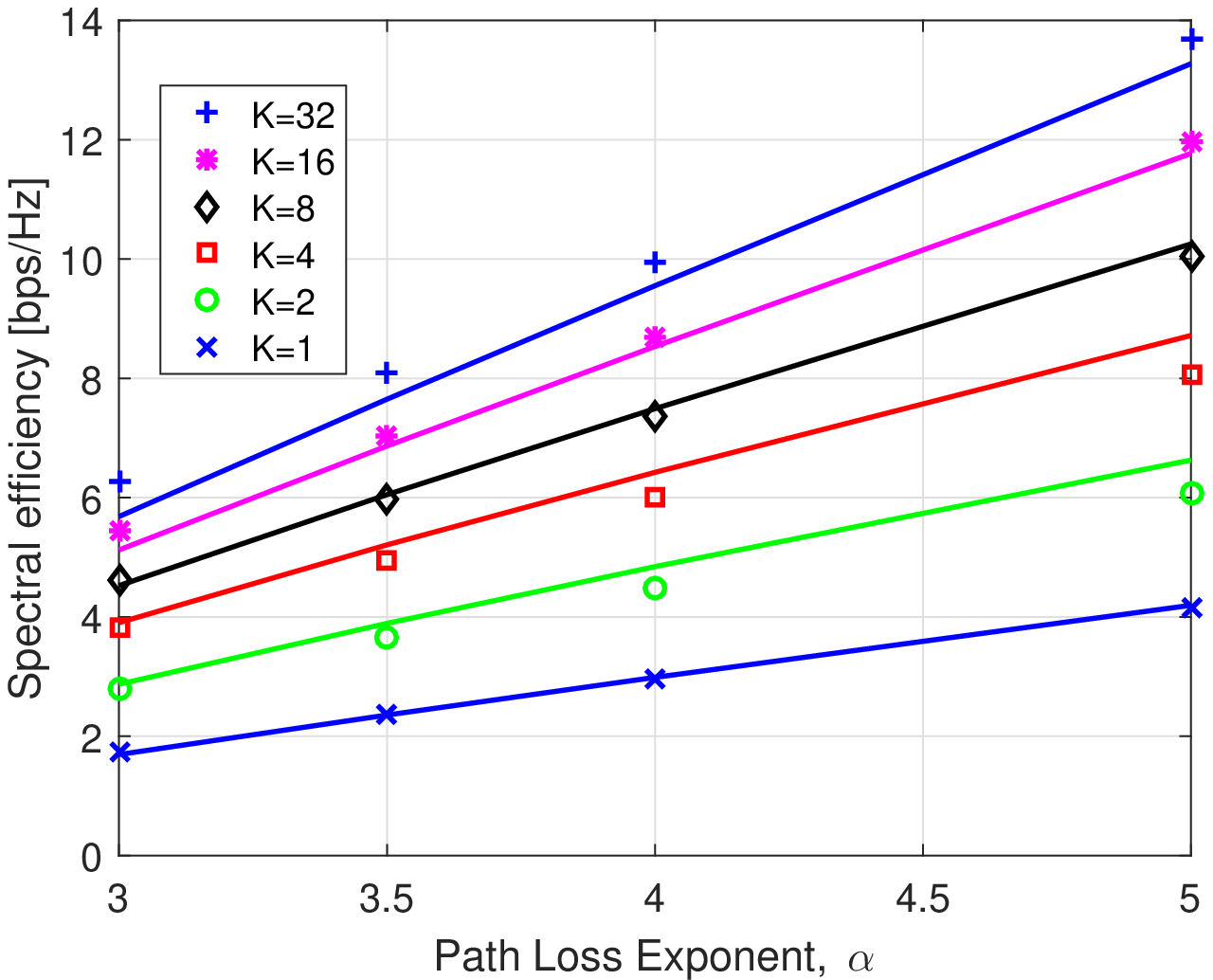}\end{center}
    \caption{Normalized spectral efficiency vs path-loss exponent. Dashed lines represent the results of Theorem \ref{Theorem:Main}. Markers depict simulation results.}
    \label{f:rate_alpha_fig}
\end{figure}
}{}

\ifstrequal{#1}{fig_hex_cell_edge}{
\begin{figure}[t]
        \begin{center}\includegraphics[scale=#2]{hex_cel_spec_eff_cell_edge_vs_L}\end{center}
    \caption{Spectral efficiency vs. number of antennas per base station for mobiles at the cell edge of a hexagonal-cell network. Dashed lines represent the results of Theorem \ref{Theorem:Main}. Solid lines depict simulated mean spectral efficiencies with error bars denoting $\pm1$ standard deviation from the mean.}
    \label{f:rate_alpha_fig}
\end{figure}
}{}

\ifstrequal{#1}{fig_hex_cell_edge_alpha}{
\begin{figure}[t]
        \begin{center}\includegraphics[scale=#2]{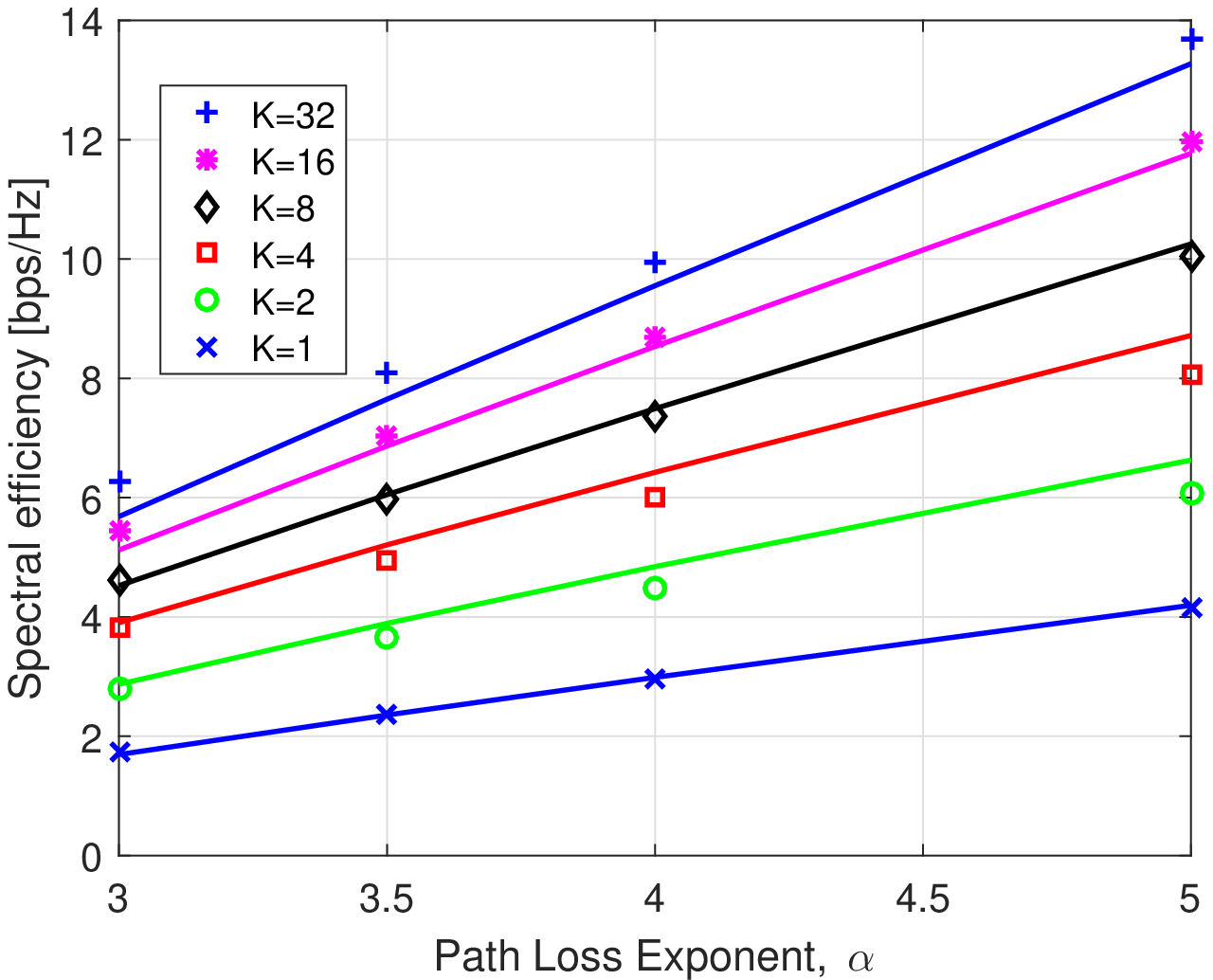}\end{center}
    \caption{Spectral efficiency vs. path-loss exponent for mobiles at the cell edge of a hexagonal-cell network. The lines represent the results of Theorem \ref{Theorem:Main} and markers depict simulated mean spectral efficiencies.}
    \label{f:rate_alpha_fig}
\end{figure}
}{}

}

\usepackage{amssymb}
\usepackage{pstricks}
\usepackage{amsmath}
\usepackage[dvips]{graphicx}
\usepackage{colortbl}
\usepackage{enumerate}
\usepackage{ifthen}
\newtheorem{theorem}{Theorem}

\newtheorem{lemma}{Lemma}

\IEEEoverridecommandlockouts

\newcommand{\Ind}[1]{
1_{\left\{#1\right\}}
}

\newcommand{\wpOne} {$  w. p.\,1 \text{ }$}

\begin{document}

\title{Uplink performance of multi-antenna cellular networks with co-operative base stations and user-centric clustering}

\author{Siddhartan Govindasamy, Itsik Bergel\thanks{\noindent Portions of this work has appeared in \cite{BSCoopICC}, Siddhartan Govindasamy (siddgov@alum.mit.edu) is with the F. W. Olin College of Engineering, Needham, MA, USA. The work of S. Govindasamy was supported in part by the United States National Science Foundation under Grant CCF-1117218.} \thanks{\noindent Itsik Bergel (bergeli@biu.ac.il) is with the Faculty of Engineering, Bar Ilan University, Ramat-Gan, Israel.}}

  \maketitle

\begin{abstract}
We consider a user-centric co-operative cellular network, where base stations (BSs) close to a mobile co-operate to detect its signal using a (joint) linear minimum-mean-square-error receiver.  The BSs are at arbitrary positions and mobiles are modeled as a planar Poisson Point Process (PPP). Combining stochastic geometry and infinite-random-matrix theory, we derive a simple expression for the spectral efficiency of this complex system as the number of antennas grows large.  This framework is applied to BS locations from PPPs and hexagonal grids, and are validated through Monte Carlo simulations. The results reveal the influence of tangible system parameters such as mobile and base-station densities, number of antennas per BS, and number of co-operating BSs on achievable spectral efficiencies. Among other insights, we find that for a given BS density and a constraint on the total number of co-operating antennas, all co-operating antennas should be located at a single BS. On the other hand, in our asymptotic regime, for the same number of co-operating antennas, if the network is limited by the area density of antennas, then the number of co-operating BSs should be increased with fewer antennas per BS.
\end{abstract}


\section{Introduction}

Co-operative processing has attracted much attention in the analysis of cellular networks, both for transmission and for reception. Such systems have been proposed in the context of cloud-radio access networks (CRAN) (see e.g. \cite{checko2015cloud}, \cite{irmer2011coordinated}) and distributed massive multiple-input multiple-output (MIMO) systems (see e.g. \cite{DistribMassiveBjornson,bjornson2015massive,truong2013viability,yin2013coordinated}), and can offer an increase in spatial diversity as signals from the antennas on multiple BSs can be used for detection. Most approaches analyzed in the literature use static clustering of base-stations (BSs) for joint processing, which
creates `super-cells' in which the clustered BSs function as a single BS with greater capabilities. While this approach provides increased spatial diversity, it  still has a problem with users that are located at the edges of the new `super-cells'. A more flexible approach is to group the clusters dynamically, such that the signals to/from the $K$ closest BSs to a given mobile user are jointly processed. Such an approach, which is a form of user-centric clustering, can completely resolve the cell-edge issue, at the cost of greater complexity in coordination.

This work aims at analyzing the uplink of a cellular network with user-centric clustering and multi-antenna BSs, where the the BSs are at arbitrary positions on a plane and the mobiles are assumed to be distributed according to a homogenous PPP (HPPP).
Note that HPPP has originally gained popularity for the modeling of ad-hoc wireless networks (e.g., \cite{HaenggiJSAC,win2009mathematical,jindal2011multi,george2015upper,rajanna2015performance}), where a uniform distribution of users  is quite intuitive. Later on, it was also used for the modeling of BSs in cellular networks (e.g., \cite{AndrewsCellular,novlan2012analytical,huang2013analytical,guo2013spatial,blaszczyszyn2013using,zhong2017heterogeneous, blaszczyszyn2015wireless,}), where it is usually compared to the hexagonal-cell model. This work mainly focuses on modeling of the mobile locations as HPPP, which is a natural extension to the intuitive assumption that the mobiles are uniformly distributed. Modeling the spatial distribution of BSs and mobiles is a powerful technique to analyze wireless cellular networks and  ``lead to remarkably precise characterizations" \cite{george2017ergodic}, and as such are very useful in the analysis of wireless networks. This work supports two of the popular models for BS locations assumed in the literature, namely the hexagonal grid and the HPPP models.
\par

We assume that the $K$ BSs closest to a given mobile cooperate to decode its signal using the linear Minimum-Mean-Square Error (MMSE) estimator, which is the optimal linear estimator to maximize the Signal-to-Interference Ratio (SIR) of the post-processed signal, under the assumption that the thermal noise is negligible.  The resulting asymptotic expressions for the achievable spectral efficiency are simple, and reveal  the dependencies between user density $\lambda$,  BS density $\lambda_b$, the number of co-operating BSs $K$, and the number of antennas per BS $L$, which  enables designers to understand the tradeoffs of increasing system resources (e.g. $K$, $L$, or $\lambda_b$)  to handle increasing densities of mobiles. Note that due to the optimality of MMSE estimator,  this result also provides an asymptotic upper bound on the performance that can be expected from all linear receivers.  Note  that for finite $L$, these bounds are approximate. We further apply these results to two commonly used models for spatial BS distribution, namely with BSs distributed either as PPPs or on a hexagonal grid, thereby illustrating the utility of our framework to different scenarios.

As was pointed out  in \cite{zhu2016stochastic}, most works on user-centric  base-station clustering have relied on numerical simulations. Analytical results on such systems are limited in the literature due to a number of technical challenges as described later in this section. Understanding the potential performance benefits of such a system is important as base-station clustering can provide very high data rates, but the cost associated with jointly processing signals received on multiple, spatially distributed BSs is potentially  very high. As such, analytical results which indicate the dependencies of different system parameters on achievable data rates are very useful.

While in the context of massive MIMO systems, it was shown that the matched filter (MF) receiver is asymptotically optimal \cite{larsson2014massive},
there is a wide range of realistic system parameters for which the MMSE receiver performs significantly better than the MF receiver \cite{hoydis2013massive}. Therefore, for practical (not so) massive MIMO systems, analyzing the MMSE receiver is very useful. Additionally,  we assume that channel information required to construct the MMSE receiver is known accurately in this work. As such these results are achievable as long as the number of antennas is not extremely large (where channel estimation and pilot contamination become significant challenges), and moreover, this result is useful as an upper bound on the  performance in situations where channel estimation errors are significant.

Analyzing the uplink of co-operative BS systems is complicated by the fact that the signals from a given mobile experience different path losses to the arrays of the co-operating BSs. While a number of works have considered the multi-antenna MMSE receiver in spatially distributed ad-hoc and  non-cooperative cellular networks (e.g., \cite{GagnonMMSE,jindal2011multi,GovCellularNetworks}), these works rely  on the fact that channel matrices for their respective models can be factored into the product a matrix representing fast fading, and a diagonal matrix containing the square-root of path-losses. Such a factorization is not possible when the signal from a given mobile experience different path-losses to the receiving array, as would be the case for a co-operative system. We address this complexity using an asymptotic analysis as the number of antennas per BS grows large. Our findings are supported by Monte Carlo simulations which show that the asymptotic results are useful even when the number of antennas is moderately large ($\sim 50$). Note that while this number is large for traditional base stations, in the near future, it is expected that base stations will be equipped with hundreds of antennas. In fact, field trials with 64 and 100 antenna base stations have already been completed (\cite{massiveMIMOWhitePaper2017,malkowsky2017world}).

A number of works have analyzed  base-station clustering systems in networks where the spatial distribution of users and BSs are considered. The majority of these works focus on the downlink, and include systems with both fixed clusters, (e.g., \cite{huang2013analytical,park2016cooperative}) and user-centric clustering (e.g. \cite{tanbourgi2014tractable,nie2016user,nigam2015spatiotemporal}). As observed in \cite{zhu2016stochastic}, most works on user centric-clustering have used  numerical simulations, e.g., \cite{papadogiannis2008dynamic} which considered zero-forcing beamformers in the uplink,  and \cite{ng2010linear} and \cite{dai2014sparse} which considered a number of precoding strategies for the downlink. In \cite{tanbourgi2014tractable}, \cite{garcia2014coordinated}, and \cite{ghods2016throughput}, the downlink of user-centric BS co-operation systems are considered analytically. Reference \cite{kim2017user} considered uplink and downlink transmissions with single antenna BSs, and 2 cooperating BSs detecting signals from a given mobile. Recently, \cite{zhu2016stochastic}  analyzed both the uplink and downlink of multi-antenna distributed processing systems using zero-forcing beamformers, and user-centric clustering. In \cite{zhu2016stochastic}, the significant complexities in analyzing  the uplink of multi-antenna base-station cooperation systems are handled using a number of approximations, e.g., on the distribution of the interference. In another recent work \cite{mohammadi2016full}, the authors consider  uplink and downlink transmissions with BSs distributed as a PPP.   For the case of uplink transmissions with the suboptimal matched-filter and zero-forcing receivers, the authors remark that statistically characterizing the SINR ``appears intractable", and use simulations to characterize this case, as well as providing an approximation for the average uplink rate. In \cite{mohammadi2016uplink}, the authors consider the optimal (i.e. MMSE), matched-filter and zero-forcing receivers using simulations, remarking that statistically characterizing the rate on the uplink ``seems to be an intractable task".

In this work on the other hand, we consider the optimal linear receiver, and derive a simple expression for the uplink spectral efficiency on a multi-antenna link which is exact asymptotically. We also provide tight approximations for the Cumulative Distribution Function (CDF) of the spectral efficiency when BSs are distributed according to a PPP, and asymptotically exact expressions for the spectral efficiency of mobiles at the cell edge of a hexagonal-grid BS deployment. The main contributions of this work can thus be summarized as follows:
\begin{itemize}
\item We provide for the first time, a closed form expression for the spectral efficiency of a multiantenna uplink in a co-operative cellular network using the optimal linear receiver.
\item This result enables us to to predict the performance of BS cooperation as a function of tangible system parameters such as base station and mobile density, number of antennas per base station and number of cooperating base stations.


\item We use our general framework to characterize the spectral efficiency of random links in a Poisson-cell network, and the cell-edge user in a hexagonal-cell network.

\end{itemize}

\section{System Model}

\begin{figure}
\begin{center}
\includegraphics[width = 3.25in]{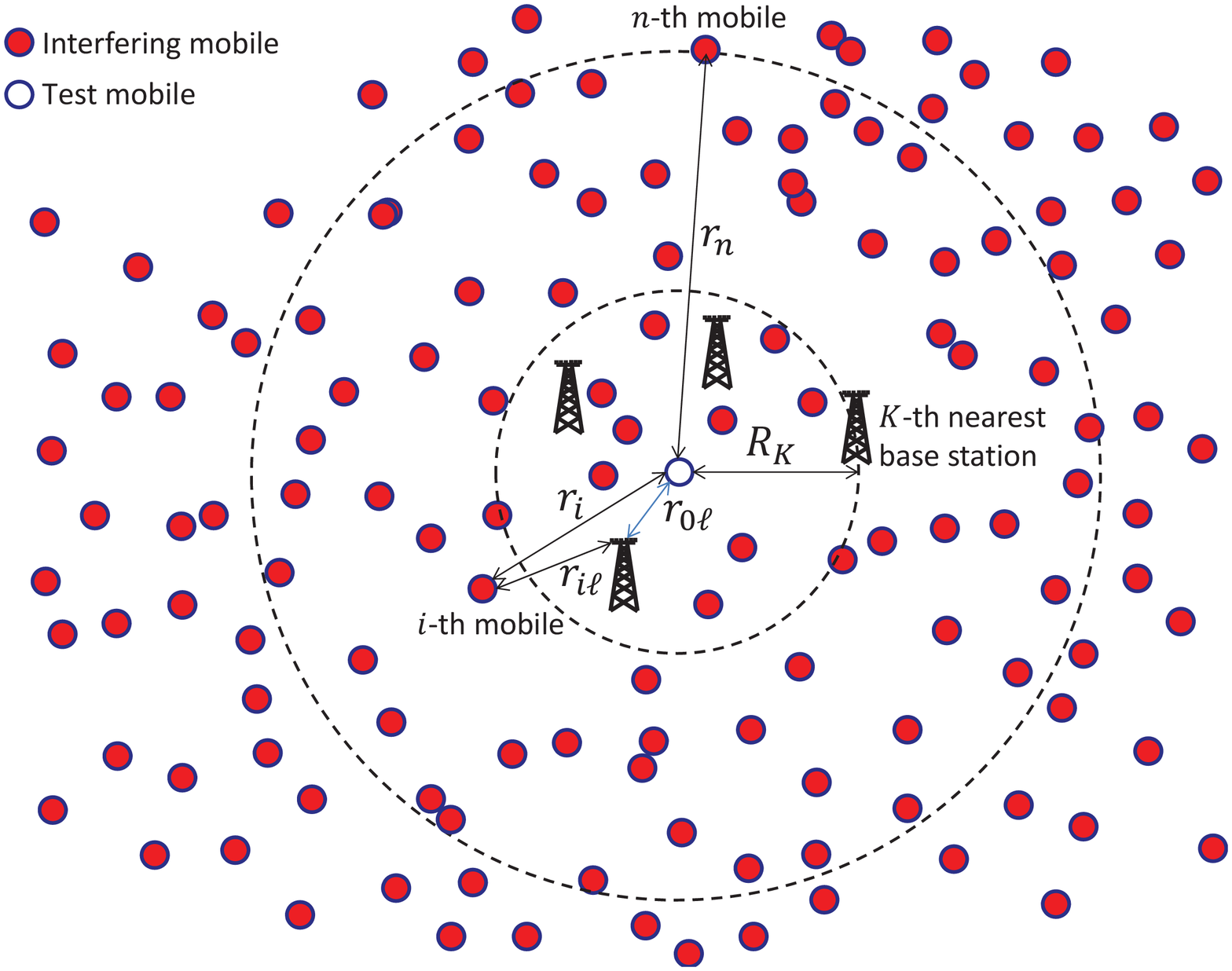}
\caption{Illustration of network with base-station cooperation.}
\label{Fig:BSCoopNetwork}
\end{center}
\end{figure}

Consider a network of BSs  at arbitrary locations on a plane, with no BS at the origin. In this work, we analyze the  link between a mobile transmitter located at the origin, and an MMSE receiver applied to signals received on the antennas of the $K$  base-stations closest to the origin. Note that later in this work, we shall apply the results derived here to two spatial BS distribution models, namely when the BSs are distributed according to an HPPP, and on a hexagonal grid.

When the BS locations are the result of a HPPP, the analyzed link is a typical link, i.e. the statistical properties of the analyzed link are equivalent to any random link in the network (for example, see the discussion on typical points in PPPs in \cite{HaenggiJSAC}). Analyzing  typical links is a widely used approach in the stochastic geometry literature. For the hexagonal-cell model, we offset the hexagonal grid such that the mobile at the origin is at the edge of 3 cells. Thus, the analysis of the link at the origin gives the performance of a cell-edge mobile.

The test mobile at the origin shall be labeled as mobile-0, whereas all other mobiles, will be co-channel interferers to the test mobile.  The other mobiles transmit simultaneously and in the same frequency band as the test mobile. The signals from the test mobile are separated spatially.  For simplicity of notation, we order both BS and mobiles according to their distance from the test mobile (although a different (random) ordering will be used in portions of the proof). Let $r_i$ denote the distance from the test mobile to the $i$-th mobile and $r_{i,j}$ denote the distance between the $i$-th mobile and the $j$-th BS (thus, our ordering results with $r_1\le r_2 \le r_3 \ldots$ and $r_{0,1}\le r_{0,2} \le r_{0,3} \ldots$).   Fig. \ref{Fig:BSCoopNetwork} illustrates this network, where we have only shown 4 of the BSs. We assume that each BS has $L$ antennas.

Ignoring the effects of thermal noise, at a given sampling time, stacking the received signals at the $L$ antennas of the $K$ co-operating BSs  into a vector $\mathbf{y}\in \mathbb{C}^{LK\times 1}$, we have
\begin{align}
\mathbf{y} = \mathbf{h}_0x + \sum_{i = 1}^\infty\mathbf{h}_i x_i\,
\end{align}
where $x_i$ is the zero-mean, unit variance transmitted symbol from the $i$-th mobile and
\begin{align}\label{e: def hi}
\mathbf{h}_i = \left(g_{i,1}\cdot r_{i, 1}^{-\alpha/2},\cdots,  g_{i,L}\cdot r_{i ,1}^{-\alpha/2},  g_{i,L+1}\cdot r_{i, 2}^{-\alpha/2},\cdots,  g_{i,2L}\cdot r_{i, 2}^{-\alpha/2}, \cdots,  g_{i,LK}\cdot r_{i, K}^{-\alpha/2}  \right)^T\,.
\end{align}
Here  $\alpha > 2$ is the path-loss exponent and $g_{i,k}$ are independent and identically distributed (i.i.d.) $\mathcal{CN}(0,1)$ random variables that represent fast fading, where the notation $\mathcal{CN}(0,1)$ indicates a zero-mean, circularly symmetric, unit-variance complex Gaussian random variable. Thus, $\mathbf{h}_i$ is a vector constructed by stacking the channel vectors between the $i$-th mobile and each of the $K$ BSs co-operating to detect the signal from the test mobile.

To avoid the use of matrices with infinite dimension, we first consider the subsystem with only the first $n$ mobiles.
The received signal in this subsystem is
\begin{align}
\mathbf{y}[n] = \mathbf{h}_0x + \sum_{i = 1}^n\mathbf{h}_i x_i.
\end{align}
Considering this $n$-mobiles subsystem, the symbol transmitted by the test mobile is estimated  using a linear MMSE estimator:
\begin{align}
\mathbf{c}[n] = \beta\left(\mathbf{H}[n]\mathbf{H}^\dagger[n]\right)^{-1}\mathbf{h}_0\,, \label{Eqn:MMSEFilt}
\end{align}
where $\mathbf{H}[n] = \left[\mathbf{h}_1 | \cdots | \mathbf{h}_n \right]$, and $\beta$ is some scale factor that does not impact the SIR as it scales the interference and signal by the same value. Note that the expression for the MMSE estimator above does not include thermal noise as our focus is on the interference-limited regime.  Thus, the SIR of the post-processed uplink signal from the test mobile is  given by
\begin{align}\label{e: SIR finite}
\mbox{SIR}[n]&=\mathbf{h}_0^\dagger \left(\mathbf{H}[n]\mathbf{H}^\dagger[n]\right)^{-1}\mathbf{h}_0.
\end{align}

Note that the matrix $\mathbf{H}[n]\mathbf{H}^\dagger[n]$ in  (\ref{e: SIR finite}) is invertible with high probability for $n>K L$. Thus, in the following, we limit the discussion to $n>K L$.  For notational convenience, we shall  define $N = K L$, which is the total number of antennas used to detect the signal from the test mobile. Note that $\mbox{SIR}[n]$ is used in the proof of the main results by taking both $n$ and $L$ to infinity.
In the proof of the main theorems, it is then shown that the final result is equivalent to $n$ being infinite (since the mobiles are from a PPP), and then taking $L$ to infinity.

The matrix $\left(\mathbf{H}[n]\mathbf{H}^\dagger[n]\right)^{-1}$ is positive definite and monotonically decreasing with $n$ (in the positive definite sense). Hence, the limit $\lim_{n\rightarrow\infty} \left(\mathbf{H}[n]\mathbf{H}^\dagger[n]\right)^{-1}$ exists, and we can define the estimator and the SIR of the original system with all the mobiles as follows
\begin{align}
\mathbf{c} &= \lim_{n\rightarrow\infty} \mathbf{c}[n]\\
\mbox{SIR} &= \lim_{n\rightarrow\infty} \mbox{SIR}[n]
\end{align}
respectively.
Note that the SIR is a random variable that depends on the locations of the mobiles and on the channel fading. Moreover, by assuming that all mobiles in the network use Gaussian codebooks, we define the spectral efficiency of the system as
\begin{align}
\eta = \log_2(1+\mbox{SIR})
\end{align}

For notational convenience, and to help interpret our results in terms of physical properties, we define   the total average power received from the test mobile by all co-operating BSs as $\mathbf{P}_K =L\sum_{k = 1}^K r_{0,k}^{-\alpha}$. In the first part of this paper, this is a deterministic quantity as the BS locations are fixed. Later, we will remove this conditioning making $\mathbf{P}_K$ random.

We study  the SIR and spectral efficiency in the asymptotic regime as $L\to\infty$. This approach provides insight into the scaling behavior of these quantities as $L$ increases, but more importantly, it provides approximations for the spectral efficiency and SIR when $L$ is fixed, but large. 

\section{ SIR and Spectral Efficiency for Arbitrary Base station Locations}
\subsection{Convergence of the normalized SIR} \label{sec:SIRLink}

In interference-limited systems, the SIR is an important quantity that determines link quality. In the following theorem, we show that with appropriate normalization, the SIR converges to a deterministic quantity as the number of antennas per BS grows large.
\begin{theorem}\label{Lemma:SIR}
The SIR in the uplink of the test mobile, satisfies the following in probability:
\begin{align}\label{e: main theorem sir}
\lim_{L\to\infty}\frac{\mbox{SIR}}{(KL)^{\alpha/2-1}\mathbf{P}_K }=\left[\frac{\,\alpha}{2\pi^2\lambda}\sin\left(\frac{2\pi}{\alpha}\right)\right]^{\frac{\alpha}{2}}
\end{align}
\end{theorem}
\noindent {\it Proof:} The proof is based on several results in large matrix theory, which are applicable when $n$ and $L$ grow with a fixed ratio. First, we shall consider a system with only the closest $n$ mobiles to the origin, and then take $n$ and $L$ to infinity. Additionally,  set $n= c L K$ with $c > 2$ which is not restrictive since the final result will have $n$ going to infinity before $L$. The applicability of this approach is guaranteed by the following lemma:
\begin{lemma}\label{L: original proof with c}
Under the conditions of Theorem \ref{Lemma:SIR}, a sufficient condition to guarantee (\ref{e: main theorem sir}) is that the following limit holds in probability:
\begin{align}\label{e: original limit with c}
\lim_{c\to\infty}\lim_{L\to\infty}\left| \frac{\mbox{SIR}[n]}{(KL)^{\alpha/2-1}\mathbf{P}_K }-\left[\frac{\,\alpha}{2\pi^2\lambda}\sin\left(\frac{2\pi}{\alpha}\right)\right]^{\frac{\alpha}{2}} \right|= 0
\end{align}
\end{lemma}
{\it Proof:} Please see Appendix \ref{App: proof of sufficinecy}.

Thus, we only need to prove (\ref{e: original limit with c}). For this, we use the following lemma which ensures that a scaled version $\mathbf{H}[n]\mathbf{H}^\dagger[n]$ in \eqref{e: SIR finite} is invertible for large $L$ with high probability.
\begin{lemma}\label{lemma:MinEval}
Assuming $c>2$, then with probability 1 (\wpOne) the smallest eigenvalue of the matrix from \eqref{e: SIR finite} is bounded from below as follows.
\begin{align}\label{e: eq lemma 2}
\lim_{L\to\infty}\gamma_{\text{min}} \left\{L^{\alpha/2-1}\mathbf{H}[n]\mathbf{H}^\dagger[n]\right\} \geq \left(\pi\lambda\right)^{\alpha/2} \left(1.5-\sqrt{2}\right)(2K)^{1-\alpha/2} \triangleq \gamma_{lb} > 0
\end{align}
\end{lemma}
{\it Proof:} Please see Appendix \ref{sec:MinEvalLemmaProof}.

Thus, for sufficiently large $L$, $L^{\alpha/2-1}\mathbf{H}[n]\mathbf{H}^\dagger[n]$ is invertible \wpOne.  Normalizing SIR$[n]$ from \eqref{e: SIR finite}, and  assuming the limit below exists, we have the following \wpOne.
\begin{align}
\lim_{L\to\infty}L^{-\alpha/2}\text{SIR}[n] &= \lim_{\zeta \to 0}\lim_{L\to\infty}\frac1L\mathbf{h}_0^\dagger \left(L^{\alpha/2-1}\mathbf{H}[n]\mathbf{H}^\dagger[n]+\zeta\mathbf{I}\right)^{-1}\mathbf{h}_0 \label{eqn: remove noise}
\end{align}

The existence of the above limit, and  its form are given in the following lemma
\begin{lemma}\label{lemma:LimitNormSIR}
 The following limit holds in probability
\begin{align}
\lim_{c\to \infty}\lim_{\zeta\to 0}\lim_{L\to\infty}\frac{1}{\mathbf{P}_K K^{\frac{\alpha}{2}-1}}\mathbf{h}_0^\dagger \left(L^{\alpha/2-1}\mathbf{H}[n]\mathbf{H}^\dagger[n]+\zeta\mathbf{I}\right)^{-1}\mathbf{h}_0 = \left[\frac{\,\alpha}{2\pi^2\lambda}\sin\left(\frac{2\pi}{\alpha}\right)\right]^{\frac{\alpha}{2}} \label{eqn: convergence to limit}
\end{align}
\end{lemma}
{\it Proof:} Please see Appendix \ref{Proof of LimitNormSIR}.

Combining \eqref{eqn: convergence to limit} with \eqref{eqn: remove noise} yields \eqref{e: original limit with c}.

Theorem \ref{Lemma:SIR} implies that for large $L$, the SIR is primarily influenced by $K$, $L$ and $\mathbf{P}_K$, and that the specific locations of interferers does not play a significant role. The physical interpretation of this property is that the large number of antennas at the BSs enable the MMSE receiver to place deep nulls in the directions of nearby interferers, reducing the dependence on the specific locations of interferers. Removing the normalization of the SIR term in Theorem \ref{Lemma:SIR},  we define the asymptotic SIR  which indicates how the SIR grows with various system parameters.
\begin{align}
\mbox{SIR}_\mathrm{asym} = \mathbf{P}_K\left[\frac{\,\alpha}{2\pi^2\lambda}\sin\left(\frac{2\pi}{\alpha}\right)\right]^{\frac{\alpha}{2}}(KL)^{\alpha/2-1}\,. \label{Eqn:SIRAsymPK}
\end{align}

\subsection{Convergence of the Spectral Efficiency} \label{sec:specEffLink}

Using Theorem \ref{Lemma:SIR}, we can characterize the asymptotic behavior of the spectral efficiency for the test link, by showing that the difference between the spectral efficiency and an asymptotic expression becomes negligible as follows.
\begin{theorem}\label{Theorem:Main}
Assuming that Gaussian codebooks are used by each mobile and neglecting thermal noise, the spectral efficiency in the uplink of the test mobile, satisfies the following:
\begin{align}\label{e: main theorem}
\lim_{L\to\infty}\left|\eta -  \eta_{\text{asym}}\right| = 0
\end{align}
where the asymptotic spectral efficiency is defined as follows
\begin{align}
\eta_{\text{asym}} = \log_2\left(1+\mathbf{P}_K\left[\frac{\,\alpha}{2\pi^2\lambda}\sin\left(\frac{2\pi}{\alpha}\right)\right]^{\frac{\alpha}{2}}(KL)^{\alpha/2-1}\right)\,. \label{eqn:AsymSpecEffDef}
\end{align}
\end{theorem}
\begin{proof}
\begin{align}
|\eta - \eta_{\mathrm{asym}}|  &= \left|\log_2(1+\mbox{SIR}) - \log_2\left(1+\left[\frac{\,\alpha}{2\pi^2\lambda}\sin\left(\frac{2\pi}{\alpha}\right)\right]^{\frac{\alpha}{2}}(KL)^{\alpha/2-1}\mathbf{P}_K\right)\right|\notag
\end{align}
\begin{align}
&=\left|\log_2\left(\frac{\frac{1}{(KL)^{\alpha/2-1}\mathbf{P}_K}+\frac{\mbox{SIR}}{(KL)^{\alpha/2-1}\mathbf{P}_K}}{\frac{1}{(KL)^{\alpha/2-1}\mathbf{P}_K}+\left[\frac{\,\alpha}{2\pi^2\lambda}\sin\left(\frac{2\pi}{\alpha}\right)\right]^{\frac{\alpha}{2}}}\right)\right|\notag
\end{align}
Since $\mathbf{P}_K$ is monotonically increasing with $L$ and $\alpha > 2$,  by Theorem \ref{Lemma:SIR} and the continuous mapping theorem, the left hand side terms in both the numerator and the denominator  in the log term above vanish in the limit, and we have \eqref{e: main theorem} in probability.
\end{proof}

Thus, for large enough $L$, the spectral efficiency is well approximated as
\begin{align}
\eta \approx \eta_\mathrm{asym}= \log_2\left(1+\mbox{SIR}_\mathrm{asym} \right)\,. \label{eqn:ApproxSpecEffAsym}
\end{align}

\section{SIR and Spectral efficiency with specific BS distributions}\label{Sec:BaseStationDist}
In this section we demonstrate the generality of the results in the previous section by applying them to two popular BS distribution models.
The results presented in the previous sections assume that BSs locations are fixed.  In subsection \ref{Sec:RandomDist}, we extend these results to base stations distributed randomly on a plane, and in section \ref{Sub: Exact SIR}, we apply it to BSs randomly distributed according to a HPPP, which is a common model for BS distributions adopted in the literature (see e.g. \cite{AndrewsCellular}, \cite{novlan2012analytical}, \cite{NovlanFracReuse} and others). In subsection \ref{sec:HexCells},  we extend these results to networks with base stations located on a hexagonal grid.

\subsection{Random spatial BS distribution}\label{Sec:RandomDist}

Suppose that the BSs are now randomly distributed in space, the probability of having a BS at the origin is zero  and with probability 1, all cells are bounded. Theorem \ref{Lemma:SIR} is applicable, when conditioned on a specific realization of the BS locations. Note from Theorem \ref{Lemma:SIR}, that the limiting normalized SIR only depends on the average received power at the co-operating antennas, $\mathbf{P}_K$. Hence, we can statistically characterize the spectral efficiency under a random spatial BS distribution by considering a random $\mathbf{P}_K$.  Let $F(x)$ denote the CDF of $\mathbf{P}_K$. Since for large $L$, Theorem \ref{Theorem:Main} provides a tight approximation for the spectral efficiency, we can approximate the CDF of the spectral efficiency as follows
\begin{align}
P(\eta \leq \tau ) &\approx P(\eta_\mathrm{asym} \leq \tau ) = F_P\left(   \frac{\left(2^\tau-1\right)K^{1-\frac{\alpha}{2}} L^{-\frac{\alpha}{2}}\texttt{}\lambda^{\frac{\alpha}{2}}}{\left[\frac{\,\alpha}{2\pi^2}\sin\left(\frac{2\pi}{\alpha}\right)\right]^{\frac{\alpha}{2}}}\right)\,.  \label{Eqn:SpecEffCDFApprox}
\end{align}
This expression provides a convenient way to statistically characterize the spectral efficiency of a link when the distribution of BSs is random.

An important performance metric for analyzing wireless systems is the outage spectral efficiency, which is the spectral efficiency that can be achieved with a given outage probability. Asymptotically, the only source of randomness in the spectral efficiency is due to $\mathbf{P}_K$. Thus for a given probability of outage $p_o$, we can find the corresponding value of $\mathbf{P}_K$ by inverting the CDF of $\mathbf{P}_K$, and denote it as follows
\begin{align}
\mathbf{P}_{K,\text{out}} = F^{-1}(p_o).
\end{align}
The spectral efficiency for this outage probability is then
\begin{align}
\log_2\left(1+\mathbf{P}_{K, \text{out}}\left[\frac{\,\alpha}{2\pi^2\lambda}\sin\left(\frac{2\pi}{\alpha}\right)\right]^{\frac{\alpha}{2}}(KL)^{\alpha/2-1}\right)\,.
\end{align}

It is important to recall that all of the analysis above was performed for a specific location of the test mobile (the origin). Thus, the distribution of $\mathbf{P}_K$ described herein should represent the BS distances as measured from this location. However, the results are much more interesting if the BS distribution is homogenous, i.e., if the distribution of distances to BSs is identical at any point. In such case, the test mobile is a typical mobile, and the derived results represent the distribution of the asymptotic SIR for any mobile. The most common model for BS locations used in the literature is to model them as a HPPP, which also satisfies the requirements outlined in this section.

\subsubsection{Poisson-distributed BSs}\label{Sub: Exact SIR}

Suppose that the BSs are distributed according to a HPPP with density $\lambda_b$ BSs per unit area. Due to Sylvnak's theorem and the ergodicity of the PPP (see e.g. \cite{Stoyan}, \cite{HaenggiJSAC}), the CDF of the spectral efficiency of the test link equals the CDF of the spectral efficiency of a randomly selected link in the network.

The spectral efficiency of the test link is found using \eqref{Eqn:SpecEffCDFApprox} with the CDF of $\mathbf{P}_K$.  The CDF of $\mathbf{P}_K$ for the HPPP model of base station distributions is characterized in the following lemma which we also expect to be useful in other contexts.
\begin{lemma} \label{lemma:PathLossSum}
The CDF of $\mathbf{P}_K$ is  $F_P(x)\triangleq P\left(  \sum_{i= 1}^Kr_{0,i}^{-\alpha} \leq \frac{x}{L}\right)$, where
\begin{align}
& P\left(  \sum_{i= 1}^Kr_{0,i}^{-\alpha} \leq x\right) = e^{-\lambda_b \pi  \left(\frac{K}{x}\right)^{\frac{2}{\alpha}}}\!\sum_{\ell = 0}^K \frac{\left(\lambda_b\pi \left(\frac{K}{x}\right)^{2/\alpha}\right)^\ell}{\ell !} +\frac{1}{\pi}\sum_{\ell = 1}^K\! {K \choose \ell} \!\! \left(-\Gamma\left(1-\frac{2}{\alpha}\right)\right)^\ell \!\sin\left(\frac{2\pi\ell}{\alpha}\right) \nonumber \\
&\times \sum_{m = 0}^\infty A_{K-\ell, m}\Gamma\left(m+\frac{2}{\alpha}\ell\right)x ^{-m-\frac{2}{\alpha}\ell} \frac{1}{K!}\left(\lambda_b\pi\right)^{\frac{1}{2}m\alpha+\ell}\Gamma\left(K-\frac{m\alpha}{2}-\ell +1, \left(\frac{K}{x}\right)^{\frac{2}{\alpha}}\lambda_b\pi \right)\label{Eqn:CDFAvgPathLossLemmaStatement}
\end{align}
The coefficients $A_{i,j}$ are defined recursively as follows
\begin{align}
A_{i,j} = \frac{1}{j}\sum_{\ell = 1}^j\left(\ell(i+1)-j\right)\left(\frac{2}{\ell ! (2-\ell \alpha)}\right)A_{i,j-\ell}  \label{Eqn:CDFCoeffs}
\end{align}
with $A_{i,0} = 1$ for $0< i\leq K$ and  $\Gamma(\cdot, \cdot)$ is the upper incomplete gamma function.
\end{lemma}

{\it Proof:}  Please see Appendix \ref{sec:PathLossSumLemmaProof}.

 While  \eqref{Eqn:SpecEffCDFApprox}  is complicated, it can be evaluated numerically with greater efficiency than running a Monte Carlo simulation for a wide range of parameters. For instance, in the numerical results described in Section \ref{Sec:NumericalResults}, the infinite sum in \eqref{Eqn:SpecEffCDFApprox} is evaluated to only 10 terms, since going beyond 10 terms did not result in an appreciable improvement in accuracy.  

\subsection{Hexagonal Cells}\label{sec:HexCells}

Here, we apply the framework developed in this paper to the hexagonal-grid model of BS locations. We analyze the spectral efficiency of a mobile located at a vertex of a hexagonal cell (the so-called cell-edge user). This case is interesting since the service to edge users is the most demanding and is a frequent test case used in cellular-network analysis. Furthermore, cell-edge mobiles can benefit the most from BS cooperation. Assume such a model with BS density $\lambda_b$ which implies a cell-edge length $d_h = \sqrt{\frac{2}{3\lambda_b\sqrt 3 }}$. We shall offset the grid from the origin such that the origin is located at the vertex of a hexagonal cell, thereby locating the test mobile at the cell-edge. With this setup, $\mathbf{P}_K$ is deterministic and  takes the following values for $K = 1,2, \cdots, 6$ (using simple geometric arguments, one can find expressions for other values of $K$ as well)
\begin{align}
\label{Eqn:PKHex}
\mathbf{P}_K=
\begin{cases}
KLd_h^{-\alpha}, & \mbox{ for } K = 1, 2, 3\\
3Ld_h^{-\alpha}+L(K-3)(2d_h)^{-\alpha}, & \mbox{ for } K = 4, 5, 6\\
\end{cases}
\end{align}

We can use this expression with Theorem \ref{Theorem:Main} to predict the spectral efficiency improvement for a cell-edge mobile using BS cooperation, over no BS cooperation. For instance, the spectral efficiency of a cell-edge mobile
with three cooperating BSs is given by
\begin{align}
\eta \approx \log_2\left(1+\left[\frac{\alpha}{2\pi^2\lambda}\sin\left(\frac{2\pi}{\alpha}\right)\right]^{\frac{\alpha}{2}}(3L)^{\alpha/2}d_h^{-\alpha}\right).
\end{align}
It is interesting to note that this spectral efficiency is equivalent to a non-cooperative BS system with three times as many antennas. Alternatively, the same spectral efficiency can be achieved with no BS cooperation, and a cell radius of $\frac{1}{\sqrt 3}d_h$, which shows how this result can be used to tradeoff between different strategies of improving service to users far away from cell towers.

\section{Insights and implications}


It is convenient to note that $\mathbf{P}_K$ is the average power received by the $K$ co-operating BSs. Thus, we define the asymptotic interference power:
\begin{IEEEeqnarray}{rCl}\label{e: limit interference power}
\sigma^2=\left[\frac{\,\alpha}{2\pi^2\lambda}\sin\left(\frac{2\pi}{\alpha}\right)\right]^{-\frac{\alpha}{2}}(KL)^{1-\alpha/2}
\end{IEEEeqnarray}
and we have $\mbox{SIR}_\mathrm{asym}=\mathbf{P}_K/\sigma^2$.
Equation \eqref{e: limit interference power}  shows that the the number of co-operating BSs, $K$, and the number of antennas per BS have the same effect on the limit interference power, i.e., the limit interference power depends only on the total number of co-operating antennas, for a particular $\mathbf{P}_K$. On the other hand, the average received power, $\mathbf{P}_K$, has a completely different dependence on $L$ and $K$, which depends on the how the BSs are distributed in space through the distances between the test mobile and its BSs, $r_{0,1}, \cdots r_{0,K}$.


To gain a better understanding, we consider a scaling of the network by a factor of $\sqrt{\lambda_b}$, with the origin as a reference point, which results in a unit-density network. Denoting the distance of the $k$-th BS from the origin in the scaled network by $\hat r_{0,k}=r_{0,k}\sqrt{\lambda_b}$,  we define:
\begin{IEEEeqnarray}{rCl}\label{e: def Sk}
\mathbf{\hat S}_K=\frac{ \mathbf{P}_K}{\lambda_b^{\alpha/2} K L}=\frac{1}{K}\sum_{k = 1}^K \hat r_{0,k}^{-\alpha}.
\end{IEEEeqnarray}
Note that $ \mathbf{\hat S}_K$ represents the average power from the $K$ nearest transmitters in a unit density network of BSs. Hence, the asymptotic SIR can be written as:
\begin{align}\label{e: final SIR formula}
\mbox{SIR}_\mathrm{asym}=\left[\frac{\,\alpha \lambda_b}{2\pi^2\lambda}\sin\left(\frac{2\pi}{\alpha}\right)\right]^{\frac{\alpha}{2}}(KL)^{\alpha/2}\mathbf{\hat S}_K
\end{align}
in which $\mathbf{\hat S}_K$ depends only on $K$, and the distances from the origin of the $K$ closest points in a unit density network.  Note that as expected for an interference-limited system, it is the relative density of BSs to mobiles, rather than the specific values of each that matters.

Writing the SIR in these different forms enables us to interpret the effects of a number of different system parameters on the SIR and spectral efficiency in the following subsections.

\subsection{Effect of number of antennas per base station}\label{Sec:NumAnts}

Equation \eqref{e: final SIR formula}  indicates that the asymptotic SIR grows as $L^{\frac{\alpha}{2}}$. Additionally, for a given $K$, if we wish to increase the density of active mobiles while keeping the same asymptotic SIR per mobile, then the number of antennas per BS should scale linearly with $\lambda$. This result is quite expected, as a similar result was obtained for cellular massive
MIMO without BS cooperation \cite{GovCellularNetworks}. Moreover, since $\alpha > 2$ for almost all practical cases, the SIR grows super-linearly with the number of antennas per base station with the MMSE receiver.

\subsection{Effect of number of cooperating BSs} \label{Sec:NumCoopBS}

Since  $\mathbf{P}_K$ is monotonically increasing with $K$, from \eqref{Eqn:SIRAsymPK}, if all other parameters are constant, the asymptotic SIR increases at a rate strictly greater than $K^{\frac{\alpha}{2}-1}$. On the other hand,  $\hat{\mathbf{S}}_K$ is monotonically decreasing with $K$. Thus, with all other parameters constant, \eqref{e: final SIR formula} indicates that the asymptotic SIR grows at a rate strictly less than $K^{\frac{\alpha}{2}}$. Thus, the SIR grows at a rate between $K^{\frac{\alpha}{2}-1}$ and $K^{\frac{\alpha}{2}}$.  From \eqref{e: final SIR formula}, observe that the asymptotic SIR grows as $\left(\frac{\lambda_b}{\lambda}\right)^{\alpha/2}$, and as $L^{\frac{\alpha}{2}}$. Therefore,  the asymptotic SIR will not decrease with mobile density $\lambda$ if the quantity $\frac{\lambda_b L K^{1-\frac{2}{\alpha}}}{\lambda}$ is kept constant by provisioning additional infrastructure (i.e. increasing $K, L,$ or  $\lambda_b$).


%
\subsection{Network optimization} \label{Sec:NetworkOptimization}

The results above give a way to evaluate the network performance and understand the role of each network parameter on the SIR and spectral efficiency in the system. In this subsection, we consider the question of network optimization, i.e., the optimal choice of the network parameters given different network constraints.

\subsubsection{Optimization given the number of co-operating antennas}
One simple measure of the network complexity is the number of co-operating antennas. In CRANs the received signals from many BSs are transferred to a cental processing center. With $K$ co-operating BSs with $L$ antennas each, $KL$ signals need to be transferred to the  processor. If $KL$ is large, this results in significant overhead on the infrastructure. Thus, the number of processed signals, $K L$ is a measure of the complexity in the infrastructure required to share signals from multiple BSs.

Trying to optimize \eqref{e: final SIR formula} with respect to $K$ and $L$ for a given complexity $N=KL$, indicates that it is preferable to increase the number of antennas per base station $L$ than the number of cooperating base stations, given that the SIR grows faster with $L$ compared to $K$, as described in Sections \ref{Sec:NumAnts} and \ref{Sec:NumCoopBS}.  Thus, \eqref{e: final SIR formula} shows that if we are only limited by the number of processed antennas, it is better to have more antennas in each BS, and fewer co-operating BSs.

\subsubsection{Optimization with limited antenna density}

But, the optimization of the previous subsection only takes into account the complexity of sharing signals from the $N$ antennas. As a result, the preferred solution requires a large number of antennas per BS and hence, a large per-base-station hardware cost. To account for this additional cost, we add a constraint on the density of antennas in the network. The antenna density is given by $\lambda_a=\lambda_b L$ antennas per unit area.

In a network with an antenna density of $\lambda_a$, we can write the asymptotic SIR, \eqref{e: final SIR formula}, as:
\begin{align}
\mbox{SIR}_\mathrm{asym}=\left[\frac{\,\alpha}{2\pi^2\lambda}\sin\left(\frac{2\pi}{\alpha}\right)\right]^{\frac{\alpha}{2}}\lambda_a^{\alpha/2}K^{\alpha/2-1}(K\mathbf{\hat S}_K).
\end{align}
Noting that $\alpha>2$ and that $K\mathbf{\hat S}_K $ is monotonically increasing with $K$, we conclude that the asymptotic SIR is monotonically increasing with $K$. Thus, it is better to  have more co-operating BSs with less antennas per BS (opposite from the conclusion in the previous subsection).

To understand the difference between the conclusions of this subsection and the previous one, we note that in both cases, the asymptotic interference power $\sigma^2$ from \eqref{e: limit interference power}, depends only on the total number of co-operating antennas. Thus, the differences result from the optimization of the average received power in each scenario ($ \mathbf{P}_K$). As $\alpha>2$, the most important factor for energy collection is the probability of having a BS which is very close to the test mobile. Thus, if the BSs density ($\lambda_b$) is fixed, increasing the number of antennas per BS increases the received energy from the nearest BS. On the other hand, when the BS density varies, it is best to distribute the BSs as much as possible in order to increase the probability of being close to the test mobile. Hence, in this case it is better to increase $\lambda_b$ and $K$ while decreasing $L$ (provided, of course that $L$ is large enough that the asymptotic expressions still hold).

\subsection{Effect of path-loss exponent}\label{Sec:SpecEffAlpha}

The effect of the path-loss exponent on the asymptotic spectral efficiency can be characterized by noting that for large $L$, \eqref{eqn:AsymSpecEffDef}  can be approximated as
\begin{align}
&\eta_{\text{asym}} \approx \left(\frac{\alpha}{2}-1\right)\log_2(KL)+\log_2(\mathbf{P}_K)+\log_2\left(\left[\frac{\,\alpha}{2\pi^2\lambda}\sin\left(\frac{2\pi}{\alpha}\right)\right]^{\frac{\alpha}{2}}\right). \label{eqn:AsymSpecEffAlpha}
\end{align}
While the right-most two terms on the RHS of \eqref{eqn:AsymSpecEffAlpha} depend on $\alpha$, their values are typically small compared to the first term on the RHS. Thus, the spectral efficiency grows approximately linearly with path-loss exponent which we verify through simulations in Section \ref{Sec:NumericalResults}.  This is consistent with the finding that the SINR increases with $\alpha$ on the downlink in \cite{tanbourgi2014tractable} .


\subsection{Optimal active mobile density} \label{Sec:OptDensity}

In the preceding, we assumed that the density of mobiles is $\lambda$, with all mobiles being active. Suppose that the true density of mobiles is much higher, and at any given time, a subset of these mobiles randomly transmit with a certain probability, such that the density of active mobiles is $\lambda$. This can be the result of random medium access control to limit congestion, for instance. We can optimize the value of $\lambda$ in order to maximize the spectral efficiency density in the network for a given outage constraint. Define the rate density for outage probability $p_o$ as follows.
\begin{align}\label{Eqn:SpecEffDensity}
\rho(p_o) = \lambda \log_2\left(1+\mathbf{P}_{K,\text{out}}\left[\frac{\,\alpha}{2\pi^2\lambda}\sin\left(\frac{2\pi}{\alpha}\right)\right]^{\frac{\alpha}{2}}(KL)^{\alpha/2-1}\right)\,.
\end{align}
Using standard calculus but tedious algebraic manipulation, we find that the density of active mobiles which maximizes the spectral efficiency  for outage probability of $p_o$ is
 \begin{align}\label{Eqn:OptDensity}
\lambda^* = \left[ -\frac{2\Omega\left( -\frac{\alpha}{2} e^{-\alpha/2}\right) } {2\Omega\left(-\frac{\alpha}{2} e^{-\alpha/2}\right)+\alpha} \right]^{2/\alpha} \left(\mathbf{P}_{K,\text{out}}(KL)^{\alpha/2-1}\right)^{\frac{2}{\alpha}}
\end{align}
where $\Omega(z)$ is the Lambert-W function and is the solution to the transcedental equation $we^w = z$. Substituting \eqref{Eqn:OptDensity} for $\lambda$ in \eqref{Eqn:SpecEffDensity}, we find that the spectral efficiency when the density of active transmissions is optimized is proportional to $\left(\mathbf{P}_{K,\text{out}}(KL)^{\alpha/2-1}\right)^{\frac{2}{\alpha}}$.

\section{Numerical Results} \label{Sec:NumericalResults}
To demonstrate the accuracy of the analysis, we conducted Monte Carlo (MC) simulations of the network topology. The simulations included $30,000$ mobiles and $3,000$ BSs in a circular area (i.e., the ratio of the densities is $\lambda_b/\lambda=0.1$. The simulation did not include white noise, and hence the actual values of the densities do not matter. All simulations were performed for a path-loss exponent of $\alpha = 4$. The simulation studied the performance of a mobile in the center of the simulation area over $10,000$ network realizations.

\FigDat{SIR_Line}{0.5}

Fig. \ref{f:SIR_Line_fig} depicts the normalized SIR, defined as
\begin{align}
\frac{\text{SIR}}{L^{\frac{\alpha}{2}-1} \mathbf{P}_K} \label{d:Normalized SIR}
\end{align}
 as a function of the number of antennas per BS for cluster sizes of $K=2, 4$ and $8$ BSs. The error bars show the range of $\pm 1$ standard deviation from the mean, according to the MC simulations. The dashed lines show the predicted values according to Theorem \ref{Lemma:SIR}. The figure shows that in all cases, the theoretical asymptotic prediction is well between the error-bars from the simulations. As the number of antennas gets large, the error-bars become very close (indicating the convergence to a deterministic quantity) and the asymptotic theoretical prediction meet the center of the bars.

\FigDat{fig_rate}{0.5}

So far, the figures showed the accuracy of the results mostly with quite high numbers of antennas. Yet, the results can be very useful even with relatively small number of antennas, if we also take into account the distribution of the received power, $\mathbf{P}_K $. Fig. \ref{f:rate_fig} compares the CDF of the asymptotic spectral efficiencies predicted by Theorem \ref{Theorem:Main}, with the CDF of the asymptotic spectral efficiencies in the MC simulations, for different parameter values as shown in the graph. In all cases, we set $\lambda = 1$, and $\alpha = 4$ for the PPP model of base stations, with different base station densities as depicted in the legend. As can be seen, with $50$ antennas per BS, the simulations deviate from the theory by only fractions of a bps/Hz. (This deviation increases to about $1$bps/Hz for  10 antennas per BS). Thus, the result of Theorem 1 gives a reasonable prediction for the achievable  performance even when the number of antennas is moderately large.

Fig. \ref{f:rate_fig} also highlights some of the conclusions that were drawn from the theoretical results  in Section \ref{Sec:NetworkOptimization}. For instance, consider the left triplet of curves, which correspond to $K L = 200$, all with $\lambda_b = 0.1$. These curves show different cases in which the total number of antennas  used to decode the signal from the test mobiles is identical. As the theory predicts, the figure demonstrates that using fewer base stations with a larger number of antennas per base station results in higher spectral efficiency if there is a restriction on the total number of antennas (200 in this case). Now consider the right triplet of curves, which depict the cases in which the total number of antennas  used to decode the signal is fixed ($K L=800$) but also the antenna density is fixed ($\lambda_b L= 10$). The graph indicates that when the density of antennas is fixed,  using a larger number of cooperating base stations with fewer antennas per base station results in higher spectral efficiencies,  as predicted by the theoretical results. Thus, these results help support the finding that depending on whether there is a restriction on the total number of cooperating antennas, or on the antenna density, the optimal deployment of resources (e.g. antennas per base station, or number of cooperating base stations) is different.

Additionally, as a point of comparison, we have plotted the empirical CDF of the spectral efficiency when a matched filter is used, for the case of 25 antennas per BS and a cluster of 32 BSs, with $\lambda_b = 0.4$. The figure indicates that the matched filter performs a lot worse than the MMSE receiver. The reason for this is that our system is interference limited, and mobiles transmit with equal power resulting in a large variation in the received powers from the mobiles due to location-dependent path loss. The large discrepancy in received signal powers severely impacts the performance of matched-filter receivers (e.g., see \cite{verdu1998multiuser}). This significant performance difference between the matched-filter and MMSE receiver indicates that the MMSE receiver is more attractive than the matched-filter not only in scenarios with moderately large numbers of antennas (as described in \cite{hoydis2013massive}), but also in distributed antenna systems with a large total number of antennas used, but with a moderately large number of antennas per BS.

\FigDat{fig_hex_cell_edge_alpha}{0.65}
Figure \ref{f:rate_alpha_fig} shows the spectral efficiency vs. path-loss exponent for the hexagonal-cell model with a relative base-station to mobile density of 0.1, $L = 50$, and values of $K$ as shown in the legend. As for the case of BSs modeled as a PPP, the simulations match the data very well when the number of antennas per BS is 50. These results help verify that the framework developed here can be used to characterize different BS distributions. Moreover, the linearity of the spectral efficiency with $\alpha$ as predicted in Section \ref{Sec:SpecEffAlpha} through \eqref{eqn:AsymSpecEffAlpha} is confirmed. In addition, observe that for larger values of $K$, the increase in spectral efficiency with $K$ grows as $\log_2(K)$ as predicted by \eqref{eqn:AsymSpecEffAlpha}.


\section{Summary and Conclusions}

In this work, we analyze the uplink of spatially distributed cellular networks with co-operative BSs which have multiple antennas. Mobiles are assumed to be distributed according to HPPPs on the plane, and BSs are at arbitrary positions.  We assume that the $K$ BSs closest in Euclidian distance to a test  mobile co-operate to detect the transmitted signal from the mobile, in the presence of co-channel interferers. Assuming Gaussian codebooks used by all mobiles, we find an asymptotic expression for the spectral efficiency on a test link given in \eqref{eqn:AsymSpecEffDef}. This result provides a simple expression for the spectral efficiency that can be achieved in co-operative cellular systems, as a function of tangible system parameters such as density of mobiles, number of co-operating BSs and number of antennas per BS. We apply these results to two specific BS distribution models, namely for BSs distributed as a HPPP and on a hexagonal grid. For the HPPP model, we find the CDF of the spectral efficiency in \eqref{Eqn:SpecEffCDFApprox} which provides the probability for a link being in outage. For the hexagonal grid model, we find expressions for the spectral efficiency of cell-edge mobiles with different numbers of cooperating BSs.

Several insights can be gained from these results. E.g., we characterize the rate at which the SIR grows with the number of base station antennas $L$, cooperating base stations $K$, and base station density $\lambda_b$.  We find that with increasing mobile density $\lambda$, at least the same per-link SIR can be supported with high probability if $K, L$ and/or $\lambda_b$ are increased in such a way that $\frac{\lambda_b}{\lambda} L K^{1-\frac{2}{\alpha}}$ is constant.  We additionally characterize the performance improvement when cell-edge mobiles use base station cooperation by providing asymptotically exact expressions for the  spectral efficiency of a cell-edge mobile. We also find the optimal density of active mobile transmissions to maximize spectral efficiency density, which balances between the interference level, and increased rate density caused by increasing the density of active links. As such, these results provide insight to system designers to make decisions on infrastructure deployment to support a given density of mobiles.1

In future work, the effects of thermal noise, and channel estimation errors should be considered. Since this work indicates the spectral efficiencies that can be achieved if these effects can be kept small enough, further study is required to determine the resources required to achieve sufficiently accurate channel estimates, and the impact of inaccurate channel estimation.

Overall, given the challenges of analytically studying the uplink of spatially distributed co-operative base-station systems which can provide high spectral efficiencies but with significant costs, we expect this result to be useful as it provides  simple expressions for the spectral efficiency which capture the effects of the main system parameters.

\section{Acknowledgements}

We thank the anonymous reviewers and editor for helpful comments which have resulted in significant improvements to this paper.

\begin{appendices}

\section{Proof of Lemma \ref{L: original proof with c}}\label{App: proof of sufficinecy}

We shall use a sandwiching argument to prove that \eqref{e: original limit with c} is a sufficient condition for \eqref{e: main theorem sir}.
First we write an upper bound to the SIR which is based on considering only the first $n$ mobiles in the network, as compared to SIR which includes an infinite number of interfering mobiles
\begin{align}\label{e: inf is higher than n}
\frac{\mbox{SIR}}{(KL)^{\alpha/2-1}\mathbf{P}_K}  &\leq \frac{\mbox{SIR}[n]}{(KL)^{\alpha/2-1}\mathbf{P}_K}
\end{align}
By applying \eqref{e: original limit with c} and recalling that we set $n=cLK$, the following holds in probability:
\begin{IEEEeqnarray}{rCl}
\lim_{L\to\infty}\frac{\mbox{SIR}}{(KL)^{\alpha/2-1}\mathbf{P}_K} &\leq& \inf_{c>2} \lim_{L\to\infty}\frac{\mbox{SIR}[n]}{(KL)^{\alpha/2-1}\mathbf{P}_K } \le \lim_{c\to\infty}\lim_{L\to\infty}\frac{\mbox{SIR}[n]}{(KL)^{\alpha/2-1}\mathbf{P}_K }
\nonumber \\&=& \left[\frac{\alpha}{2\pi^2\lambda}\sin\left(\frac{2\pi}{\alpha}\right)\right]^{\frac{\alpha}{2}}. \label{eqn:SIRUBLimit}
\end{IEEEeqnarray}
where the first inequality holds because \eqref{e: inf is higher than n} holds for every $c$.

\newcommand{\emm}{{u}}

Next, consider the MMSE estimator for a network with only $n$ mobile nodes, $\mathbf{c}[n]$. If this estimator is applied to the network with all the mobiles, the resulting SIR, which we define as  $\mbox{\underline{SIR}}[n]$, will be a lower bound to the SIR with the full MMSE estimator. Hence,
\begin{align}\label{e: SIR finite2}
\text{SIR} \geq &\mbox{\underline{SIR}}[n] = \lim_{\emm\to\infty}\frac{\left|\mathbf{c}[n]^\dagger\mathbf{h}_0\right|^2}{\mathbf{c}[n]^\dagger \left(\mathbf{H}[\emm]\mathbf{H}^\dagger[\emm]\right)\mathbf{c}[n]} =\lim_{\emm\to\infty}\frac{\left|\mathbf{c}[n]^\dagger\mathbf{h}_0\right|^2}{\mathbf{c}[n]^\dagger \left(\mathbf{H}[n]\mathbf{H}^\dagger[n]+\mathbf{\Delta}_n[\emm]\mathbf{\Delta}_n^\dagger[\emm]\right)\mathbf{c}[n]}
\notag \\
&=\lim_{\emm\to\infty}\frac{\left|\mathbf{c}[n]^\dagger\mathbf{h}_0\right|^2}{\mathbf{c}[n]^\dagger \mathbf{H}[n]\mathbf{H}^\dagger[n]\mathbf{c}[n]\left(1+\frac{\mathbf{c}^\dagger[r]\mathbf{\Delta}_n[\emm]\mathbf{\Delta}_n^\dagger[\emm]\mathbf{c}[n]}{\mathbf{c}[n]^\dagger \mathbf{H}[n]\mathbf{H}^\dagger[n]\mathbf{c}[n]}\right)}
=\frac{\mbox{SIR}[n]}{1+\frac{\lim_{\emm\to\infty}\mathbf{c}^\dagger[n]\mathbf{\Delta}_n[\emm]\mathbf{\Delta}_n^\dagger[\emm]\mathbf{c}[n]}{\mbox{SIR}[n]}}
\end{align}
where $\mathbf{\Delta}_n[\emm] = \left[\mathbf{h}_{n+1} | \cdots | \mathbf{h}_\emm \right]$. Using again $n=cKL$, we can write:
\begin{IEEEeqnarray}{rCl}\label{eqn: vanishing denominator}
\lim_{L\to\infty}\frac{\text{SIR}}{(KL)^{\alpha/2-1} \mathbf{P}_K}&\ge&\sup_{c>2}\lim_{L\to\infty}\frac{(KL)^{1-\alpha/2} \mathbf{P}_K^{-1}\mbox{SIR}[n]}{1+\frac{\lim_{\emm\to\infty}\mathbf{c}^\dagger[n]\mathbf{\Delta}_n[\emm]\mathbf{\Delta}_n^\dagger[\emm]\mathbf{c}[n]}{\mbox{SIR}[n]}} \nonumber\\
&\ge &\lim_{c\to\infty}\lim_{L\to\infty}\frac{(KL)^{1-\alpha/2} \mathbf{P}_K^{-1}\mbox{SIR}[n]}{1+\frac{(KL)^{1-\alpha/2} \mathbf{P}_K^{-1}\lim_{\emm\to\infty}\mathbf{c}^\dagger[n]\mathbf{\Delta}_n[\emm]\mathbf{\Delta}_n^\dagger[\emm]\mathbf{c}[n]}{(KL)^{1-\alpha/2} \mathbf{P}_K^{-1}\mbox{SIR}[n]}}
\nonumber\\
&= &\frac{\left[\frac{\alpha}{2\pi^2\lambda}\sin\left(\frac{2\pi}{\alpha}\right)\right]^{\frac{\alpha}{2}}}{1+\frac{\lim_{c\to\infty}\lim_{L\to\infty}(KL)^{1-\alpha/2} \mathbf{P}_K^{-1}\lim_{\emm\to\infty}\mathbf{c}^\dagger[n]\mathbf{\Delta}_n[\emm]\mathbf{\Delta}_n^\dagger[\emm]\mathbf{c}[n]}{\left[\frac{\alpha}{2\pi^2\lambda}\sin\left(\frac{2\pi}{\alpha}\right)\right]^{\frac{\alpha}{2}}}}.
\end{IEEEeqnarray}
Recalling that $\mathbf{P}_K /L=\sum_{k = 1}^K r_{0,k}^{-\alpha}<\infty$, the following lemma shows that the  denominator of the RHS of \eqref{eqn: vanishing denominator} converges to unity.
\begin{lemma}\label{lemma:ConvergenceOfQuadraticForm}
\begin{align}
& \lim_{c\to\infty}\lim_{L\to\infty}\lim_{\emm\to\infty} L^{-\alpha/2} \mathbf{c}^\dagger[n]\mathbf{\Delta}_n[\emm]\mathbf{\Delta}_n[\emm]^\dagger \mathbf{c}[n] = 0 \;\;\;\;\text{in probability.
} \label{Eqn:DenConv}
\end{align}
\end{lemma}
{\it Proof:} Please see Appendix \ref{Sec:ConvergenceOfQuadraticFormProof}.

Hence the following holds in probability,
\begin{align}
\lim_{L\to\infty}\frac{\text{SIR}}{(KL)^{\alpha/2-1} \mathbf{P}_K} &\geq   \left[\frac{\alpha}{2\pi^2\lambda}\sin\left(\frac{2\pi}{\alpha}\right)\right]^{\frac{\alpha}{2}}\,.
\end{align}
Comparing to  \eqref{eqn:SIRUBLimit} we observe that the upper and lower bounds converge to the same constant in probability. Hence we
conclude that \eqref{e: main theorem sir} holds in probability.

\section{Proof of Lemma \ref{lemma:MinEval}}\label{sec:MinEvalLemmaProof}

First lets define $M = 2KL$. As $c>2$, we have $n = cKL > M$, and by the Weyl inequality (see e.g. \cite{HornJohnson}),
\begin{align}
\gamma_{\text{min}} \left(M^{\alpha/2-1}\mathbf{H}[n]\mathbf{H}^\dagger[n]\right)\label{eqn:MinEvalWeylInit}\ge\gamma_{\text{min}} \left(M^{\alpha/2-1}\mathbf{H}[M]\mathbf{H}^\dagger[M]\right)
\end{align}
whereby,
\begin{align}
\mathbf{H}[M]\mathbf{H}^\dagger[M] &= \sum_{i = 1}^M \mathbf{h}_i \mathbf{h}_i^\dagger =  \sum_{i = 1}^M \left(r_{M+1}+R_{K+1}\right)^{-\alpha} \mathbf{g}_i\mathbf{g}_i^\dagger + \left[\sum_{i = 1}^M (\mathbf{h}_i \mathbf{h}_i^\dagger-\left(r_{M+1}+R_{K+1}\right)^{-\alpha}\mathbf{g}_i\mathbf{g}_i^\dagger)\right] \label{eqn:SubCovarianceMatrix}
\end{align}
where $\mathbf{g}_i = (g_{i1}, g_{i2}, \cdots, g_{iN})^T$. Recall the definition of $\mathbf{h}_i$ from \eqref{e: def hi}.
Since $r_{i,j}< r_{M+1}+R_{K+1} $ for $i = 1, 2, \cdots, M$ and $j = 1, 2, \cdots, K$ the matrix in the brackets on the RHS of \eqref{eqn:SubCovarianceMatrix} is non-negative definite \wpOne. Hence, by the Weyl inequality (see e.g. \cite{HornJohnson})
\begin{align}\label{eqn:MinEvalWeyl}
\gamma_{\text{min}} \left\{\mathbf{H}[M]\mathbf{H}^\dagger[M]\right\}
\ge   \gamma_{\text{min}} \left\{\sum_{i = 1}^M \left(r_{M+1}+R_{K+1}\right)^{-\alpha} \mathbf{g}_i\mathbf{g}_i^\dagger\right\}
\end{align}
From \cite{BaiSmallestEval}, the following is known to hold \wpOne:
\begin{align}\label{eqn:SmallestEvalStdCov}
\lim_{M\to\infty}  \gamma_{\text{min}} \left\{\sum_{i = 1}^M  \frac1M\mathbf{g}_i\mathbf{g}_i^\dagger\right\} = \left(1-\sqrt{\frac{1}{2}}\right)^2.
\end{align}
Additionally, since $r_{M+1}$ is a scaled $\chi^2$ random variable with $2M+2$ degrees of freedom (see e.g. \cite{Mathai} Equation 2.4.4), the following holds in the mean-square sense
\begin{align}
\lim_{M\to\infty} \frac{{\pi\lambda}r_{M+1}^2}{M+1} = 1\,.
\end{align}
Additionally, since $R_{K+1}$ is finite, the following holds in the mean-square sense,
\begin{align}\label{eqn:OuterPathLossConvergence}
\lim_{M\to\infty}\left(r_{M+1}+R_{K+1}\right)^{-\alpha}\left(\frac{M}{\pi\lambda}\right)^{\alpha/2} = 1
\end{align}
Since mean-square convergence and convergence \wpOne both imply convergence in probability, combining \eqref{eqn:SmallestEvalStdCov} and \eqref{eqn:OuterPathLossConvergence}, we have
\begin{align}
\lim_{M\to\infty} M^{\alpha/2-1}\gamma_{\text{min}}\left\{\left(r_{M+1}+R_{K+1}\right)^{-\alpha} \sum_{i = 1}^M  \mathbf{g}_i\mathbf{g}_i^\dagger\right\} = \left(\pi\lambda\right)^{\alpha/2} \left(1.5-\sqrt{2}\right)
\end{align}
in probability.
Combining this expression with \eqref{eqn:MinEvalWeyl} and \eqref{eqn:MinEvalWeylInit}  yields the following in probability
\begin{align}
\lim_{L\to\infty} \gamma_{\text{min}} \left\{M^{\alpha/2-1}\mathbf{H}[n]\mathbf{H}^\dagger[n]\right\}\ge\left(\pi\lambda\right)^{\alpha/2} \left(1.5-\sqrt{2}\right).
\end{align}
Recalling also that $M=2 K L$ leads to \eqref {e: eq lemma 2} and completes the proof.

\section{Proof of Lemma \ref{lemma:LimitNormSIR}} \label{Proof of LimitNormSIR}
To prove this lemma, we use a prior result from \cite{BaiSilversteinMIMOCDMA}. In order to use this result, we define  $q_{i,j}$ as i.i.d. random variables taking values of $\pm 1$ with equal probability, and $\bar{g}_{i,j} = g_{i,j}/q_{i,j}$. Note that $\bar{g}_{i,j}$ are still i.i.d $\mathcal{CN}(0,1)$ random variables. We thus have
\begin{align}\label{e: def hbar}
\mathbf{h}_i = &\left(\bar{g}_{i,1}\cdot  q_{i, 1} \cdot  r_{i, 1}^{-\alpha/2},\cdots,  \bar{g}_{i,L}\cdot q_{i ,1}\cdot r_{i ,1}^{-\alpha/2},  \bar{g}_{i,L+1} \cdot q_{i,2}\cdot r_{i,2}^{-\alpha/2},\cdots,\right. \nonumber \\
 &\;\;\;\;\;\;\;\;\;\;\;\;\;\;\;\;\;\;\;\;\;\;\;\;\;\;\;\;\;\;\;\;\;\;\;\;\left. \bar{g}_{i,2L}\cdot q_{i,2}\cdot r_{i,2}^{-\alpha/2}, \cdots,  \bar{g}_{i,LK}\cdot q_{i,K}\cdot r_{i,K}^{-\alpha/2}  \right)^T\,.
\end{align}

Next we define the empirical distribution function (e.d.f.) of   $N^{\alpha/2}|q_{i,j}r_{i,j}^{-\alpha/2}|^2 = N^{\alpha/2}r_{i,j}^{-\alpha}$ as the function $H_N(\tau_1, \cdots, \tau_K)$ below
\begin{align}
H_N(\tau_1, \cdots, \tau_K) = \frac{1}{n} \sum_{i = 1}^n 1_{\{N^{\alpha/2}r_{i,1}^{-\alpha} \leq \tau_1, \cdots, N^{\alpha/2}r_{i,K}^{-\alpha} \leq \tau_K\}}.
\end{align}
Hence, $H_N(\tau_1, \cdots, \tau_K)$ measures the proportion of the $n$ mobiles for which the quantities $N^{\alpha/2}r_{i,1}^{-\alpha} \leq \tau_1, \cdots, N^{\alpha/2}r_{i,K}^{-\alpha}$ are less than or equal to $\tau_1, \cdots, \tau_K$, respectively. From Corollary 1.2 of  \cite{BaiSilversteinMIMOCDMA} if the following two conditions are true,
 \begin{enumerate}
\item  $H_N(\tau_1, \cdots, \tau_K)$ converges \wpOne to a $K$ dimensional, non-random function $H(\tau_1, \tau_2, \cdots, \tau_K)$\,,
\item For all $\ell \neq \ell'$, and all positive $\nu_1, \cdots, \nu_K$, the following holds \wpOne,
 \begin{align}
\lim_{n\to\infty} \frac{1}{n-1} \sum_{i = 1}^n \frac{q_{i,\ell}{q_{i, \ell'}N^{\alpha/2}r_{i,\ell}^{-\alpha/2}r_{i,\ell'}^{-\alpha/2}}}{1+\sum_{j= 1}^K \nu_j N^{\alpha/2}r_{i,j}^{-\alpha}} = 0\,,
\end{align}
 \end{enumerate}
then, with probability 1,
\begin{align}
\lim_{L\to\infty} \frac{1}{L }\mathbf{h}_0^\dagger \left(\frac1L N^{\alpha/2}\mathbf{H}[n]\mathbf{H}^\dagger[n]+\zeta\mathbf{I}\right)^{-1}\mathbf{h}_0 
&= \sum_{i = 1}^K r_{0,i}^{-\alpha} a_i\,, \label{Eqn:partialSIRLimit}
\end{align}
where $a_i$ is given uniquely by the set of $a_i$s satisfying the following equations for $i = 1, 2, \cdots, K$.
\begin{align}
a_i = \frac{1}{\frac{n}{L}E\left[\frac{\tau_i}{1+\sum_{j = 1}^K a_j \tau_j}\right]+\zeta}. \label{eqn:implicitsol}
\end{align}
The $\tau_i$ terms are random variables with joint CDF $H(\tau_1, \tau_2, \cdots, \tau_K)$.

The convergence of the e.d.f. of $N^{\alpha/2}r_{i,j}^{-\alpha}$ required for condition 1) above is proved in the first part of  the following lemma. The second part is used to evaluate \eqref{eqn:implicitsol}.
\begin{lemma}\label{lemma:EDF}
As $N\to\infty$, $H_N(\tau_1, \cdots, \tau_K)$, converges \wpOne to  a joint CDF \\
$H(\tau_1, \tau_2, \cdots, \tau_K)$. Moreover, if $\tau_1, \tau_2, \cdots, \tau_K$ are random variables with joint CDF $H(\tau_1, \cdots, \tau_K)$, the CDF of  $\tau = \frac{ \tau_1+\tau_2+\cdots+ \tau_K}{K}$,  equals $G(\tau)$ given below
\begin{align}
G(\tau)=1-\frac{\pi\lambda}{c}\left(\tau\right)^{-\frac{2}{\alpha}}\Ind{\left(\frac{\pi\lambda_p}{c}\right)^{\alpha/2}
< \tau }. \label{eqn24}
\end{align}
\end{lemma}
{\it Proof:} Please see Appendix \ref{sec:EDFDistLemmaProof}.

Condition 2)  can be proved as follows.
 \begin{align}
 \left|\frac{1}{n-1} \sum_{i = q}^n \frac{q_{i,\ell}{q_{i, \ell'}N^{\alpha/2}r_{i,\ell}^{-\alpha/2}r_{i,\ell'}^{-\alpha/2}}}{1+\sum_{j= 1}^K \nu_j N^{\alpha/2}r_{i,j}^{-\alpha}}\right| &\leq \left|\frac{1}{n-1} \sum_{i = q}^n \frac{q_{i,\ell}{q_{i, \ell'}\left(\max_{\ell}\{r_{i,\ell}\}\right)^{-\alpha}}}{\sum_{j= 1}^K \nu_j r_{i,j}^{-\alpha}}\right| \leq \nonumber \\
\left|\frac{1}{n-1} \sum_{i = q}^n \frac{q_{i,\ell}{q_{i, \ell'}\left(\max_{\ell}\{r_{i,\ell}\}\right)^{-\alpha}}}{\min_{k}\{\nu_k\}\sum_{j= 1}^K  r_{i,j}^{-\alpha}}\right| &\leq
   \frac{1}{\min_{k}\{\nu_k\}}\left|\frac{1}{n-1} \sum_{i = q}^n q_{i,\ell}q_{i, \ell'}\right| \label{Eqn:MinWeightSum}
\end{align}
\eqref{Eqn:MinWeightSum} holds because the summation in the denominator of the term on the left-hand side (LHS) is over all $j$, and hence must include the term $\max_\ell \{r_{i,\ell}^{-\alpha}\}$. As $n,N\to\infty$, for $\ell \neq \ell'$, the RHS above converges to zero \wpOne  by the strong law of large numbers.

To find  $a_i$, observe that  by symmetry, the fact that \eqref{eqn:implicitsol} has a unique solution leads to $a_1 = a_2 = \cdots = a_K = a$. Thus, we can write  \eqref{eqn:implicitsol} as
\begin{align}
a &= \frac{1}{\frac{n}{L}E\left[\frac{\tau_i}{1+a\sum_{j = 1}^K \tau_j}\right]+\zeta} = 
\frac{1}{\frac{n}{L}E\left[\frac{\tau}{1+aK\tau}\right]+\zeta} \label{Eqn:CompactExpressiona}
\end{align}
where $\tau$ is as defined in Lemma \ref{lemma:EDF}. Using the symmetry of the $a_i$ again in  \eqref{Eqn:partialSIRLimit} gives:\begin{align}
\lim_{L\to\infty} \frac1L\mathbf{h}_0^\dagger \left(\frac1LN^{\alpha/2}\mathbf{H}[n]\mathbf{H}^\dagger[n]+\zeta\mathbf{I}\right)^{-1}\mathbf{h}_0 &= a \sum_{i = 1}^K r_{0,i}^{-\alpha}\label{Eqn:SIRConvergenceImplicit} 
\end{align}
Recalling that $N = LK$ and $c = n/N$, if we rewrite  \eqref{Eqn:CompactExpressiona} as
\begin{align}
(aK)cE\left[\frac{\tau}{1+(aK)\tau}\right]+a\zeta = 1\,,
\end{align}
we can use a prior result from \cite{govindasamy13Correlated} where the uplink of non-co-operative BSs was analyzed (i.e. the case when $K=1$). From \cite{govindasamy13Correlated}, it is known that the following holds
\begin{align}
\lim_{c\to \infty}\lim_{\zeta\to 0} \;aK = \left[\frac{\,\alpha}{2\pi^2\lambda}\sin\left(\frac{2\pi}{\alpha}\right)\right]^{\frac{\alpha}{2}}\,.
\end{align}
Taking the limits as  $\zeta \to 0$  followed by $c\to\infty$ of \eqref{Eqn:SIRConvergenceImplicit} and substituting the previous expression,
and using $N = KL$,
\begin{align}
\lim_{c\to \infty}\lim_{\zeta\to 0} \lim_{L\to\infty} \frac{K^{1-\frac{\alpha}{2}}}{L}\mathbf{h}_0^\dagger \left(L^{\frac{\alpha}{2}-1}\mathbf{H}[n]\mathbf{H}^\dagger[n]+\zeta\mathbf{I}\right)^{-1}\mathbf{h}_0 &= \left[\frac{\,\alpha}{2\pi^2\lambda}\sin\left(\frac{2\pi}{\alpha}\right)\right]^{\frac{\alpha}{2}} \sum_{i = 1}^K r_{0,i}^{-\alpha}. \label{Eqn:SIRConvergenceImplicit3}
\end{align}
Substituting the definition of $\mathbf{P}_K$ and rearranging completes the proof.


\section{Proof of Lemma \ref{lemma:PathLossSum}}\label{sec:PathLossSumLemmaProof}

To find the CDF of $ \sum_{i= 1}^K r_{0,i}^{-\alpha}$, we first condition on the event that $r_{0,(K+1)} = R_{K+1}$. Since the BSs are distributed according to a HPPP, the $K$ BSs closest to the test mobile are distributed with uniform probability in the interior of the disk of radius $R_{K+1}$, centered at the origin. Let $\tilde{r}_{01},\tilde{r}_{02}, \cdots, \tilde{r}_{0K}$, be the distances of the $K$ BSs closest to the test mobile, in \emph{random} order.  
Note that $R_{K+1}^\alpha \tilde{r}_{0i}^{-\alpha}$ are Pareto distributed random variables with shape parameter $2/\alpha$. For $\alpha > 2$, the sum of $K$ i.i.d. Pareto distributed random variables has the following CDF  for $x \geq K$,
\begin{align}
 &P\left( \left. \sum_{i= 1}^KR_{K+1}^\alpha\tilde{r}_{0i}^{-\alpha} \leq x\right| R_{K+1}\right)=1\nonumber\\
&\!+\frac{1}{\pi}\sum_{\ell = 1}^K {K \choose \ell} \left(-\Gamma\left(1-\frac{2}{\alpha}\right)\right)^\ell \!\sin \left(\frac{2\pi}{\alpha}\ell\right)\sum_{m = 0}^\infty A_{K-\ell, m}\Gamma\left(m+\frac{2}{\alpha}\ell\right)x ^{-m-\frac{2}{\alpha}\ell}\,, \label{Eqn:CDFAvgPathLossCond1}
\end{align}
and zero otherwise \cite{blum1970sums}.
From  \eqref{Eqn:CDFAvgPathLossCond1}, we have
\begin{align}
 &P\left( \left. \sum_{i= 1}^K r_{0i}^{-\alpha} \leq x\right| R_{K+1}\right) = P\left( \left. \sum_{i= 1}^K\tilde{r}_{0i}^{-\alpha} \leq x\right| R_{K+1}\right) = P\left( \left. \sum_{i= 1}^KR_{K+1}^{\alpha}\tilde{r}_{0i}^{-\alpha} \leq x R_{K+1}^{\alpha}\right| R_{K+1}\right)\nonumber\\
&=1+\frac{1}{\pi}\sum_{\ell = 1}^K \!{K \choose \ell} \!\!\left(-\Gamma\left(1-\frac{2}{\alpha}\right)\right)^\ell \!\sin \left(\frac{2\pi}{\alpha}\ell\right)\sum_{m = 0}^\infty A_{K-\ell, m}\Gamma\!\left(m+\frac{2}{\alpha}\ell\right)\!\!\left(xR_{K+1}^{\alpha}\right) ^{-m-\frac{2}{\alpha}\ell}\, \label{Eqn:CDFAvgPathLossCond2}
\end{align}
for $xR_{K+1}^{\alpha} \geq K$ and zero otherwise. Since $R_{K+1}$ is the distance of the $(K+1)$st nearest neigbor in a HPPP, its PDF is known (see e.g. \cite{Mathai}). To remove the conditioning on $R_{K+1}$ in  \eqref{Eqn:CDFAvgPathLossCond2} with respect to $R_{K+1}$, we first take the expectation of the term $R_{K+1} ^{-m\alpha -2\ell}\Ind{xR_{K+1}^\alpha \geq K}$,
\begin{align}
E\left[R_{K+1} ^{-m\alpha -2\ell}\, \Ind{xR_{K+1}^\alpha \geq K}\right] &= \int_{\left(\frac{K}{x}\right)^{\frac{1}{\alpha}}}^\infty R_{K+1} ^{-m\alpha -2\ell} \frac{2(\lambda_b \, \pi R_{K+1}^2)^{K+1}}{ R_{K+1} \, K!}e^{-\lambda_b \, \pi R_{K+1}^2}\, dR_{K+1}
\end{align}
Making the substitution $q = R_{K+1}^2\lambda_b\pi$, $dq = dR_{K+1} 2R_{K+1}\lambda_b\pi $ and solving yields
\begin{align}
E\left[R_{K+1} ^{-m\alpha -2\ell}\, \Ind{xR_{K+1}^\alpha \geq K}\right] =\frac{\left(\lambda_b\pi\right)^{m\alpha/2+\ell}}{K!}\Gamma\left(K-\frac{m\alpha}{2} - \ell +1, \left(\frac{K}{x}\right)^{\frac{2}{\alpha}}\lambda_b \pi\right)
\end{align}
Additionally, applying  the CDF of $R_{K+1}$  which can be found e.g. in \cite{Mathai}, we have
\begin{align}
E\left[ \Ind{xR_{K+1}^\alpha \geq K}\right] &= e^{-\lambda_b \pi \left(\frac{K}{x}\right)^{2/\alpha}}\sum_{\ell = 0}^K \frac{\left(\lambda_b\pi  \left(\frac{K}{x}\right)^{2/\alpha}\right)^\ell}{\ell !}\,.
\end{align}

Using the previous two expressions to take the expectation  of the conditional CDF given in \eqref{Eqn:CDFAvgPathLossCond2} with respect to $R_{K+1}$, yields \eqref{Eqn:CDFAvgPathLossLemmaStatement}.


\section{Proof of Lemma \ref{lemma:ConvergenceOfQuadraticForm}}\label{Sec:ConvergenceOfQuadraticFormProof}

First, we define $\epsilon[n,\emm]=\mathbf{c}^\dagger[n]\mathbf{\Delta}_n[\emm]\mathbf{\Delta}_n^\dagger[\emm]\mathbf{c}[n]$ and note that:
\begin{align}
\epsilon[n,\emm]&=\|\mathbf{\Delta}_n^\dagger[\emm]\mathbf{c}[n]\|^2 =\sum_{m=n+1}^\emm |\mathbf{h}_m^\dagger \mathbf{c}[n] |^2
\notag \\
\lim_{\emm\to\infty}\epsilon[n,\emm] &= \epsilon[n]\triangleq \sum_{m=n+1}^\infty |\mathbf{h}_m^\dagger \mathbf{c}[n] |^2\label{eqn:epsilon_def}
\end{align}
With $\mathbf{h}_{m,k}^\dagger$ and $\mathbf{c}_k[n]$ denoting the sub-vectors associated with the $k$-th BS, we have
\begin{IEEEeqnarray}{rCl}
|\mathbf{h}_m^\dagger \mathbf{c}[n]|^2=\sum_{k=1}^K |\mathbf{h}_{m,k}^\dagger \mathbf{c}_k[n]|^2.\label{eqn:channel_vec_decomp}
\end{IEEEeqnarray}

Given $\mathbf{c}_k[n]$, the product $\mathbf{h}_{m,k}^\dagger \mathbf{c}_k[n]$ is a complex Gaussian random variable with zero mean and a variance of $  r_{m, k}^{-\alpha} \|\mathbf{c}_k[n]\|^2$ which is statistically independent of any other product in the summations of \eqref{eqn:epsilon_def} and \eqref{eqn:channel_vec_decomp}. Thus, we can define \begin{IEEEeqnarray}{rCl}
q_{m,k}=\frac{\mathbf{h}_{m,k}^\dagger \mathbf{c}_k[n]}{r_{m, k}^{-\alpha/2} \|\mathbf{c}_k[n]\|}
\end{IEEEeqnarray}
and note that given $\mathbf{c}_k[n]$, these are iid random variables with complex Gaussian distribution of zero mean and unit variance. With this notation, \eqref{eqn:epsilon_def} becomes
\begin{IEEEeqnarray}{rCl}
\epsilon[n] = \sum_{m=n+1}^\infty \sum_{k=1}^K \left|r_{m, k}^{-\alpha/2} \|\mathbf{c}_k[n]\|\cdot q_{m,k}\right|^2\,.
\end{IEEEeqnarray}

Since $r_{m,k}>|r_m-R_k|$ whenever $r_m>R_k$, we can bound $\epsilon[n]$ as follows:
\begin{align}
&\epsilon[n] \leq \sum_{m=n+1}^\infty \sum_{k=1}^K \left| (r_m -R_K)^{-\alpha/2} \|\mathbf{c}_k[n]\|\cdot q_{m,k}\right|^21_{\{r_m > R_K\}} \nonumber \\
& \;\;\;\; + \sum_{m=n+1}^\infty \sum_{k=1}^K \left|r_{m, k}^{-\alpha/2} \|\mathbf{c}_k[n]\|\cdot q_{m,k}\right|^21_{\{r_m \leq R_K\}}.
 \label{eqn:EpsilonBound}
\end{align}
The second term in the RHS of \eqref{eqn:EpsilonBound} is non-zero only if $r_{n+1} \leq R_k$. Thus, it converges to zero w.p. 1:
\begin{align}\label{eqn:LimitOfInnerCircleTerm}
\lim_{n\rightarrow\infty}&\Pr\left\{\sum_{m=n+1}^\infty \sum_{k=1}^K \left|r_{m, k}^{-\alpha/2} \|\mathbf{c}_k[n]\|\cdot q_{m,k}\right|^21_{\{r_m \leq R_K\}} > 0\right\} \leq \lim_{n\rightarrow\infty}\Pr(r_{n+1} \leq R_k)= 0.
\end{align}
Turning to the first term in the RHS of \eqref{eqn:EpsilonBound}, we use   \cite{george2013ergodic} that
shows that
\begin{align}
\lim_{n\to\infty}   n^{\alpha/2-1} \sum_{m=n+1}^\infty  (r_m -R_K)^{-\alpha} |q_{m,k}|^21_{\{r_m > R_K\}} = \frac{2 (\pi \lambda)^{\alpha/2}}{(\alpha-2)}.
\end{align}
Thus,
\begin{align}\label{eqn:LimitOfOuterCircleTerm}
\lim_{n\to\infty}   n^{\alpha/2-1} \sum_{m=n+1}^\infty  (r_m -R_K)^{-\alpha} \|\mathbf{c}_k[n]\|^2\cdot |q_{m,k}|^2 1_{\{r_m > R_K\}} =\|\mathbf{c}_k[n]\|^2 \cdot\frac{2 (\pi \lambda)^{\alpha/2}}{(\alpha-2)}
\end{align}

Combining \eqref{eqn:EpsilonBound}, \eqref{eqn:LimitOfInnerCircleTerm} and  \eqref{eqn:LimitOfOuterCircleTerm}, and recalling that $\|\mathbf{c}[n]\|^2 =\sum_{k=1}^K \|\mathbf{c}_k[n]\|^2 $, implies the following in probability
\begin{align}
\lim_{n\to\infty}n^{\alpha/2-1}\frac{\epsilon[n] }{\|\mathbf{c}[n]\|^2}
& \leq  \frac{2 (\pi \lambda)^{\alpha/2}}{(\alpha-2)}\,.\label{eqn:LB_convergence_partial1}
\end{align}
Hence, we have:
\begin{align}\label{eq: almost last limit in lemma}
   \lim_{L\to\infty} L^{-\alpha/2} \epsilon[n] &= \lim_{L\to\infty} L^{1-\alpha}(cK)^{1-\alpha/2} \|\mathbf{c}[n]\|^2 \left(n^{\alpha/2-1}\frac{\epsilon[n] }{\|\mathbf{c}[n]\|^2} \right) \nonumber\\ &\leq \lim_{L\to\infty}L^{1-\alpha}(cK)^{1-\alpha/2} \|\mathbf{c}[n]\|^2\cdot\frac{2 (\pi \lambda)^{\alpha/2}}{(\alpha-2)}.
\end{align}
 Since $ \|\mathbf{c}[n]\|^2 = \mathbf{h}_0^\dagger\left(\mathbf{H}[n]\mathbf{H}^\dagger[n]\right)^{-2} \mathbf{h}_0$,
applying Lemma \ref{lemma:MinEval} we have:
\begin{align}
\lim_{L\to\infty} \frac1L\mathbf{h}_0^\dagger\left(L^{\alpha/2-1}\mathbf{H}[n]\mathbf{H}^\dagger[n]\right)^{-2}\mathbf{h}_0 &\leq\lim_{L\to\infty}
\frac{1}{\left(\gamma_{\text{min}}\left\{L^{\alpha/2-1}\mathbf{H}[n]\mathbf{H}^\dagger[n]\right\}\right)^2}\frac1L\mathbf{h}_0^\dagger\mathbf{h}_0 \nonumber \\
&\le \frac{\sum_{k = 1}^K r_{0,k}^{-\alpha}}{\left(\pi\lambda\right)^{\alpha} \left(1.5-\sqrt{2}\right)^2(2K)^{2-\alpha}}
\label{eqn:QuadraticFormBound}
\end{align}
Substituting the previous expression into \eqref{eq: almost last limit in lemma}:
\begin{align}
& \lim_{L\to\infty}\lim_{\emm\to\infty} L^{-\alpha/2} \mathbf{c}^\dagger[n]\mathbf{\Delta}_n[\emm]\mathbf{\Delta}_n[\emm]^\dagger \mathbf{c}[n] \leq (cK)^{1-\alpha/2} \frac{2 (\pi \lambda)^{\alpha/2}\sum_{k = 1}^K r_{0,k}^{-\alpha}}{\left(\pi\lambda\right)^{\alpha} \left(1.5-\sqrt{2}\right)^2(2K)^{2-\alpha}(\alpha-2)},\notag \end{align}
and since $\alpha > 2$, we have \eqref{Eqn:DenConv}.

\section{Proof of Lemma \ref{lemma:EDF}}\label{sec:EDFDistLemmaProof}
Let $\tilde{r}_i$, $i = 1, 2, \cdots, n$ be the distances from the origin of the $n$ closest mobiles to the origin, in random order, and $\tilde{r}_{i,j}$ be the distance of the corresponding mobile from the $j$-th BS, for $i = 1, 2,\cdots, n$ and $j = 1,2,\cdots, K$.
Because the mobiles are distributed according to a HPPP, conditioned on the distance of the $(n+1)$th mobile from the origin, $r_{n+1}$,  the $n$ mobiles closest to the origin are distributed with uniform probability in a disk of radius $r_{n+1}$ centered at the origin.  Since the $K$ BSs closest to the origin are at most a distance $R_K$ from the origin, conditioned on $r_{n+1}$, and recalling that $n/N = c$,  we have
\begin{align}
\Pr(N^{\alpha/2}\tilde{r}_{i,j}^{-\alpha} \leq \tau) \leq & \Pr\left(N^{\alpha/2}\left(\tilde{r}_{i}+R_K\right)^{-\alpha} \leq \tau\right) =\Pr(\tilde{r}_i \geq \tau^{-1/\alpha} N^{1/2} -R_K)\nonumber \\
 &=1-\frac{\left(\tau^{-\frac{1}{\alpha}}  -\frac{R_K}{\sqrt N}\right)^{2} n}{c\,r_{n+1}^2}1_{\left\{\left(\tau^{-\frac{1}{\alpha}}  -\frac{R_K}{\sqrt N}\right)^{-\alpha} > \left(\frac{n}{c r_{n+1}^2}\right)^{\frac{\alpha}{2}}\right\}} \label{eqn:CDFBeforeLimit}
\end{align}

Since $2\pi\lambda r^2_{n+1}$ is a $\chi^2$ random variable with $2n+2$ degrees of freedom (see e.g. \cite{Mathai} Equation 2.4.4), it can be written as the sum of the squares of $2n+2$ i.i.d. standard normal random variables, whose average equals $\frac{2\pi\lambda r_{n+1}^2}{2n+2}$. Hence, by the strong law of large numbers, we have
\begin{align}
\lim_{n\to\infty} \frac{r_{n+1}^2}{n} = \frac{1}{\pi\lambda}\,,\label{eqn:ConvegenceOfNormChiSquare}
\end{align}
\wpOne. Substituting \eqref{eqn:ConvegenceOfNormChiSquare} into \eqref{eqn:CDFBeforeLimit}, we get the following  \wpOne
\begin{align}
\lim_{N\to\infty}\Pr(N^{\alpha/2}\tilde{r}_{i,j}^{-\alpha} \leq \tau) &\leq 1-\frac{\tau^{-2/\alpha} \pi\lambda}{c}1_{\left\{\left(\frac{\pi\lambda}{c}\right)^{\alpha/2} < \tau\right\}}.\label{eqn:CDFLimitUB}
\end{align}
For a lower bound, we can write
\begin{align}
\Pr(N^{\alpha/2}\tilde{r}_{i,j}^{-\alpha} \leq \tau) &\geq  \Pr\left(\left.N^{\alpha/2}\left(\tilde{r}_{i}-R_K\right)^{-\alpha} \leq \tau \right|\tilde{r}_{i}>R_K \right)\Pr(\tilde{r}_{i}>R_K)\nonumber \\
 &\;\;\;\;\;\;\;\;\;\;\;\;\;\;\;\;\;\;\;\;\;+ \Pr\left(\left.N^{\alpha/2}\tilde{r}_{i,j}^{-\alpha} \leq \tau \right|\tilde{r}_{i}>R_K \right)\Pr(\tilde{r}_{i}\leq R_K)\label{eqn:CDFLB}
\end{align}
Since as $n\to\infty$, $\Pr(\tilde{r}_{i}\leq R_K) \to 0$ \wpOne, the second term on the RHS of \eqref{eqn:CDFLB} goes to zero \wpOne. Hence we only need to consider the first term on the RHS in the limit.
\begin{align}
\Pr\left(\left.N^{\frac{\alpha}{2}}\left(\tilde{r}_{i}-R_K\right)^{-\alpha} \leq \tau \right|\tilde{r}_{i}>R_K \right)  &=  1-\frac{\left(\tau^{-\frac{1}{\alpha}}  +\frac{R_K}{\sqrt N}\right)^{2} n}{c\, r_{n+1}^2}1_{\left\{\left(\tau^{-1/\alpha}  + \frac{R_K}{\sqrt N}\right)^{-\alpha} > \left(\frac{n}{c r_{n+1}^2}\right)^{\frac{\alpha}{2}}\right\}}
\end{align}
Hence, we get the following \wpOne,
\begin{align}
\lim_{N\to\infty}\Pr\left(\left.N^{\alpha/2}\left(\tilde{r}_{i}-R_K\right)^{-\alpha} \leq \tau \right|\tilde{r}_{i}>R_K \right)  \geq 1-\frac{\tau^{-2/\alpha} \pi\lambda}{c}1_{\left\{\left(\frac{\pi\lambda}{c}\right)^{\alpha/2} < \tau\right\}}\label{eqn:CDFLimitLB}
\end{align}
Since as $n,N\to\infty$, $\Pr(r_i > R_K) \to 1$, taking the limit of \eqref{eqn:CDFLB}  and substituting \eqref{eqn:CDFLimitLB} yields the following \wpOne.
\begin{align}
\lim_{N\to\infty}\Pr(N^{\alpha/2}\tilde{r}_{i,j}^{-\alpha} \leq \tau) &\geq 1-\frac{\tau^{-2/\alpha} \pi\lambda}{c}1_{\left\{\left(\frac{\pi\lambda}{c}\right)^{\alpha/2} < \tau\right\}}\label{eqn:CDFLimitLBFin}
\end{align}
which together with \eqref{eqn:CDFLimitUB} implies the following \wpOne
\begin{align}
\lim_{N\to\infty}\Pr(N^{\alpha/2}\tilde{r}_{i,j}^{-\alpha} \leq \tau) &= 1-\frac{\tau^{-2/\alpha} \pi\lambda}{c}1_{\left\{\left(\frac{\pi\lambda}{c}\right)^{\alpha/2} < \tau\right\}}\triangleq F_r(\tau)\,.\label{eqn:CDFLimitFin}
\end{align}
By Proposition 5.25 of \cite{Karr}, we  have the following \wpOne, for some joint CDF $F_{\mathbf{r}}(\tau_1, \cdots, \tau_K)$,
\begin{align}
\lim_{N\to\infty} \Pr(N^{\alpha/2}\tilde{r}_{i,j}^{-\alpha} \leq \tau_1, \cdots, N^{\alpha/2}\tilde{r}_{i,j}^{-\alpha} \leq \tau_K )  = F_{\mathbf{r}}(\tau_1, \cdots, \tau_K). \label{Eqn:CDFLimVec}
\end{align}

Since the joint e.d.f. of the $\tilde{r}_{i,j}^{-\alpha}$ i.e. $H_N(\tau_1, \cdots, \tau_K)$ converges \wpOne to $  \Pr(N^{\alpha/2}\tilde{r}_{i,j}^{-\alpha} \leq \tau_1, \cdots, N^{\alpha/2}\tilde{r}_{i,j}^{-\alpha} \leq \tau_K )$  by the strong law of large numbers, combined with  \eqref{Eqn:CDFLimVec}, yields
\begin{align}
\lim_{N\to\infty} H_N(\tau_1, \cdots, \tau_K) = F_{\mathbf{r}}(\tau_1, \cdots, \tau_K) = H(\tau_1, \cdots, \tau_K) \;\;{w .p. 1}. \label{Eqn:JointEDFConvFin}
\end{align}

Thus, the first part of the lemma is proved. For the second part, consider the following:
\begin{align}
&\Pr\left\{N^{\frac{\alpha}{2}}r_{i,1}^{-\alpha} + \cdots + N^{\frac{\alpha}{2}}r_{i,K}^{-\alpha}  \leq K\tau | r_{i} > R_K \right\} =\Pr\left\{r_{i,1}^{-\alpha} + \cdots + r_{i,K}^{-\alpha} \leq K\tau N^{\frac{-\alpha}{2}} | r_{i} > R_K \right\}\nonumber \\
&\geq \Pr\left\{(r_{i}-R_K)^{-\alpha} + \cdots + (r_{i}-R_K)^{-\alpha} \leq K\tau N^{\frac{-\alpha}{2}} | r_{i} > R_K \right\}\nonumber \\
&= \Pr\left\{(r_{i}-R_K)^{-\alpha} \leq \tau N^{\frac{-\alpha}{2}} | r_{i} > R_K \right\}\
\end{align}
Following the steps used to derive \eqref{eqn:CDFLimitUB}, we get the following \wpOne.
\begin{align}
&\lim_{N\to\infty} \Pr\left\{N^{\frac{\alpha}{2}}r_{i,1}^{-\alpha} + \cdots + N^{\frac{\alpha}{2}}r_{i,K}^{-\alpha}\leq K\tau  | r_{i} > R_K  \right\} \geq  1-\frac{\tau^{-2/\alpha} \pi\lambda}{c}1_{\left\{\left(\frac{\pi\lambda}{c}\right)^{\alpha/2} < \tau\right\}}
\end{align}
Using similar steps we can get
\begin{align}
&\Pr\left\{N^{\frac{\alpha}{2}}r_{i,1}^{-\alpha} + \cdots + N^{\frac{\alpha}{2}}r_{i,K}^{-\alpha}  \leq K\tau  | r_{i} > R_K\right\} \leq \Pr\left\{(r_{i}+R_K)^{-\alpha} \leq \tau N^{\frac{-\alpha}{2}} | r_{i} > R_K \right\}
\end{align}
Following the steps used to derive \eqref{eqn:CDFLimitLB}, we get the following \wpOne.
\begin{align}
&\lim_{N\to\infty} \Pr\left\{N^{\frac{\alpha}{2}}r_{i,1}^{-\alpha} + \cdots + N^{\frac{\alpha}{2}}r_{i,K}^{-\alpha}\leq K\tau  | r_{i} > R_K  \right\} \leq  1-\frac{\tau^{-2/\alpha} \pi\lambda}{c}1_{\left\{\left(\frac{\pi\lambda}{c}\right)^{\alpha/2} < \tau\right\}}
\end{align}
Since as $n,N\to\infty$, $r_{i} > R_K$ \wpOne, we conclude the following \wpOne.
\begin{align*}
&\lim_{N\to\infty} \Pr\left\{\frac1K \left(N^{\frac{\alpha}{2}}r_{i,1}^{-\alpha} + \cdots + N^{\frac{\alpha}{2}}r_{i,K}^{-\alpha}\right)\leq \tau  | r_{i} > R_K  \right\} = 1-\frac{\tau^{-2/\alpha} \pi\lambda}{c}1_{\left\{\left(\frac{\pi\lambda}{c}\right)^{\alpha/2} < \tau\right\}} = G(\tau).
\end{align*}
Using the approach used to prove \eqref{Eqn:JointEDFConvFin}, we find that the e.d.f. of $\frac1K \left(N^{\frac{\alpha}{2}}\!r_{i,1}^{-\alpha} + \cdots + N^{\frac{\alpha}{2}}\!r_{i,K}^{-\alpha}\right)$ converges to $G(\tau)$ \wpOne . Since the joint e.d.f. of $N^{\frac{\alpha}{2}}\!r_{i,1}^{-\alpha},\cdots,N^{\frac{\alpha}{2}}\!r_{i,K}^{-\alpha}$ converges \wpOne to $H(\tau_1, \cdots , \tau_K)$, the CDF of  $\frac1K \left(\tau_1 + \cdots +\tau_K \right)$, where the $\tau_i$ terms have joint CDF $H(\tau_1, \cdots , \tau_K)$, must converge \wpOne to $G(\tau)$.

\end{appendices}

\bibliographystyle{IEEEtran}
\bibliography{IEEEabrv,main}

\end{document}